\documentclass[10pt,journal]{IEEEtran}
\usepackage[cmex10]{amsmath}
\usepackage[short]{optidef}
\usepackage{color}
\usepackage{commath}
\usepackage{algorithm}
\usepackage{algorithmic}
\usepackage{textcomp}
\usepackage{gensymb}
\usepackage{multirow}
\usepackage{multicol}
\usepackage{mathtools}

\ifCLASSINFOpdf
   \usepackage[pdftex]{graphicx}
\graphicspath{ {figures/} }
\else
\fi
\usepackage{graphicx}
\usepackage[font=footnotesize,caption=false,subrefformat=parens,labelformat=parens]{subfig}

\usepackage{cite}
\usepackage{wrapfig}
\usepackage{amssymb}
\usepackage{amsthm}
\usepackage{pifont}
\newcommand{\cmark}{\ding{51}}%
\newcommand{\xmark}{\ding{55}}%

\usepackage{siunitx}

\DeclarePairedDelimiter\floor{\lfloor}{\rfloor}  

\usepackage{eucal}
\usepackage{enumitem}
\newtheorem{prop}{Lemma}
\newtheorem{thm}{Theorem}
\newtheorem{cor}{Remark}


\newcommand\blfootnote[1]{%
		\begingroup
		\renewcommand\thefootnote{}\footnote{#1}%
		\addtocounter{footnote}{-1}%
		\endgroup
}

\title{Mobility State Detection of Cellular-Connected UAVs based on Handover Count Statistics}
\author{Md Moin Uddin Chowdhury, Priyanka Sinha, Kim Mahler, and
\.{I}smail G\"{u}ven\c{c}
\thanks{M.M.U. Chowdhury, P. Sinha, and \.{I}. G\"{u}ven\c{c} are with the Department of Electrical and Computer Engineering, North Carolina State University, Raleigh, NC 27606 (e-mail:~\{mchowdh,psinha2,iguvenc\}@ncsu.edu).}
\thanks{K. Mahler is with Voltela Inc., 11 E Loop Rd, Suite 381 New York, NY 10044  (e-mail:~kim.mahler@voltela.com).}}
\begin{document}
\pdfoutput=1
\maketitle
\begin{abstract}
To ensure reliable and effective mobility management for aerial user equipment (UE), estimating the speed of cellular-connected unmanned aerial vehicles (UAVs) carries critical importance since this can help to improve the quality of service of the cellular network. The 3GPP LTE standard uses the number of handovers made by a UE during a predefined time period to estimate the speed and the mobility state efficiently. In this paper, we introduce an approximation to the probability mass function of handover count (HOC) as a function of a cellular-connected UAV's height and velocity, HOC measurement time window, and different ground base station (GBS) densities. Afterward, we derive the Cramer-Rao lower bound (CRLB) for the speed estimate of a UAV, and also provide a simple biased estimator for the UAV's speed which depends on the GBS density and HOC measurement period. Interestingly, for a low time-to-trigger (TTT) parameter, the biased estimator turns into a minimum variance unbiased estimator (MVUE). By exploiting this speed estimator, we study the problem of detecting the mobility state of a UAV as low, medium, or high mobility as per the LTE specifications. Using CRLBs and our proposed MVUE, we characterize  the accuracy improvement in speed estimation and mobility state detection as the GBS density and the HOC measurement window increase. Our analysis also shows that the accuracy of the proposed estimator does not vary significantly with respect to the TTT parameter.    


\end{abstract}
\section{Introduction}
\blfootnote{This work is supported by NSF grants CNS-1453678 and CNS-1910153. Part of this work was presented at IEEE Signal Processing Advances in Wireless Communications (SPAWC) Workshop in 2020~\cite{chowdhurySPAWC2020}.}

Thanks to their flexibility in deployment and their low production/operational costs, using unmanned aerial vehicles (UAVs) have attracted significant interest for a wide range of commercial and civilian applications in recent years~\cite{rui1,geraci_2018,halim_mobility}. For taking full advantage of UAV deployments, beyond visual line of sight (BVLOS) operations are of critical importance where UAVs can fly autonomously without direct human control. Most importantly, such autonomous UAV flights will enable future commercial applications such as air taxis and air ambulances by means of dedicated air corridors for safe and efficient operations of aerial vehicles~\cite{sinha2022wireless}. 

Thanks to the recent efforts from Federal Aviation Agency (FAA), it is now allowed to fly a UAV autonomously in BVLOS scenarios by using the commercial cellular service \cite{VzWlteDrone}. FAA has also signed a three-year agreement with Verizon-owned Skyward to experiment with the use of cellular-connected UAVs \cite{FAADrone}. A Research Transition Team is in place between the FAA,  National Aeronautics and Space Administration (NASA), and industry to coordinate the UAS Traffic Management (UTM) initiative to enable safe visual and BVLOS UAV flights in low-altitude airspace of under $400$ feet above the ground level. Efforts are underway in Europe~\cite{5Gdrones} with 5G experimentation to make safe BVLOS UAV flights possible. Indeed, existing cellular networks can be a strong candidate for operating autonomous UAVs in BVLOS scenarios, in which a UAV acts as an aerial user equipment (UE) and can maintain reliable communication for safety and control purposes with the ground base stations (GBSs)~\cite{geraci_2018}. 

For maintaining reliable and seamless connection quality at the cellular-connected UAVs, effective mobility management by minimizing handover failures, radio link failures (RLFs), as well as unnecessary handovers is critically important. However, due to being served by sidelobes of the GBS that provide lower antenna gains, a UAV might be connected with a GBS located far from it~\cite{geraci_2018,ramy_antenna}. This phenomenon, in turn, makes the reference signal received power (RSRP) based mobility management of cellular-connected UAVs extremely challenging. Moreover, the network operators are unlikely to sacrifice the performance of the ground UE performance for supporting aerial UEs.

Speed estimation of a ground/aerial UE can play an important role in effective mobility management. This information, in turn, can help efficient resource scheduling, load balancing, and energy efficiency enhancements~\cite{arvind_ho}. Especially, due to the patchy signal coverage of GBS in the sky, a high UAV speed indicates that the UAV in interest will be associated with a GBS for a brief amount of time. These challenges motivate the need for UAV or cell-specific handover parameter optimization based on the effective UAV speed estimation. Moreover, UAVs flying at high altitudes suffer from high interference stemming from nearby GBS due to the near free-space path-loss trend in the GBS-to-UAV link~\cite{lin_sky}. Estimating the mobile UAV speed will enable the GBSs to coordinate among themselves for leveraging the inter-cell interference coordination scheme as done for ground UEs in~\cite{david2012}. While the global positioning system (GPS) can be used to estimate speed, GPS receivers consume a significant amount of power~\cite{arvind_ho}. Furthermore, GPS coverage in dense urban canyons may be unreliable. In such a case, the number of HOs made by the UAV within a certain measurement window due to the patchy coverage can be utilized to estimate its speed~\cite{3gpp.36.331}. The estimated speed can be also used to determine the mobility state (low, medium, or high) of the flying UAV. Such kind of mobility state detection (MSD) analysis will help to integrate UAVs reliably into the existing cellular networks. For instance, a UAV flying at high speed will need faster handover processing to ensure that outage does not happen due to rapidly changing coverage patterns. Thus the handover parameters for low and high UAV velocities can also be tuned properly to reduce handovers or minimize RLFs~\cite{turkka_moblility}. Hence, by obtaining the UAV speed efficiently, it is possible to avoid outages and improve mobility performance significantly. 
Note that we did not consider any handover improvement technique in this paper, e.g. approaches similar to those discussed in~\cite{lin_field, denmark_uav_test,moin_ICC,lin_mobility}. If HO enhancements for UAVs are deployed, our proposed speed and MSD approaches can still be used after certain modifications and parameter optimization.

\begin{table*}[t]
\centering
\caption{{Literature review on UAV Speed Estimation.}} 
\scalebox{0.989}{
\begin{tabular}{p{0.4cm} p{3.8cm} p{2.8cm} p{1.2cm} p{1.0cm} p{1.2cm} p{1.0cm} p{1.2cm} p{1.2cm}} \hline
{\textbf{Ref.}} & {\textbf{Goal}} & {\textbf{Antenna pattern}} & {\textbf{Ground reflection}} & {\textbf{Speed estimation}} &{\textbf{MSD}} & {\textbf{Aerial UE}}  & {\textbf{Impact of HO parameters}} & {\textbf{Impact of UAV height}}\\ \hline

\cite{Ramy_walid_ICC2020} & Analyzing cellular-connected UAVs considering $3$D antenna radiation & directional, array & \xmark & \xmark & \xmark & \cmark & \xmark & \cmark\\ \hline

\cite{ramy_coverage} & Providing reliable connectivity and mobility support for UAVs & directional, array   & \xmark & \xmark & \xmark & \cmark & \xmark & \cmark\\ \hline

\cite{galkin2020} & Intelligent GBS association for UAVs based on network information  & directional, array   & \xmark & \xmark & \xmark & \cmark & \xmark & \cmark\\ \hline


\cite{mahdiAari_RL2020} & Reducing disconnectivity time, handover rate, and energy consumption of UAV & directional, array & \xmark & \xmark & \xmark & \cmark & \xmark & \xmark\\ \hline

\cite{moin_ICC} & Serving both ground users
and UAVs simultaneously in a co-channel sub-6 GHz network & directional, array   & \xmark & \xmark & \xmark & \cmark  & \xmark & \xmark\\ \hline

\cite{xingqin_rl_2019} & Ensuring robust wireless connectivity and mobility support for UAVs & directional, array & \xmark & \xmark & \xmark  & \cmark & \xmark & \xmark\\ \hline

\cite{lin2020a2g} & Maximizing aircraft user throughput by tuning tilting angles and inter GBS distance  & directional, array & \cmark & \xmark & \xmark & \cmark & \xmark & \xmark\\ \hline

\cite{simran_GR} & Serving both ground users
and UAVs simultaneously in a co-channel mmWave network & directional, single   & \xmark & \xmark & \xmark & \cmark & \xmark & \cmark\\ \hline
\cite{arvind_ho} & Estimating ground user speed, CRLB, and mobility state  & omnidirectional, single   & \xmark & \cmark & \cmark & \xmark  & \xmark & \xmark\\ \hline
\cite{chowdhurySPAWC2020} & Estimating UAV speed and CRLB & directional, array   & \xmark & \cmark & \xmark & \cmark & \xmark & \xmark\\ \hline
\textbf{This work} & Estimating UAV speed, CRLB, and mobility state  & directional, array   & \cmark & \cmark & \cmark & \cmark & \cmark & \cmark\\ \hline
\end{tabular}}
\label{tab:lit_review}
\end{table*}

In our previous work~\cite{chowdhurySPAWC2020}, we provided preliminary results considering the impact of antenna radiation to obtain the handover count (HOC) PMFs while we did not tackle the problem of MSD based on the estimated speed. Moreover, the speed estimation analysis was done for a fixed UAV height, antenna configuration, and handover related parameters. In contrast, in this paper, we approximate the HOC PMF using Gaussian distribution while considering both antenna patterns and ground reflection which provides a better representation for the RSRP values stemming from the down-tilted GBS antennas~\cite{chowdhury2021ensuring}. We also report the HOC PMF for different UAV heights, antenna element numbers, and handover parameters such as measurement gap and time-to-trigger (TTT). Afterward, we propose a speed estimator based on the HOC and then obtain the MSD based on this proposed estimator. Our key contributions can be summarized as follows.
\begin{itemize}[leftmargin=*]
    \item We first introduce a novel and efficient HOC based UAV speed estimation technique while considering realistic GBS antenna radiation pattern~\cite{chowdhury2021ensuring}, handover scenario~\cite{3gpp.36.331}, and ground reflection~\cite{chowdhury2021ensuring}. Due to the complex antenna configurations and ground reflection-based path loss model, it is difficult to obtain closed-form expression of HOC PMF. Hence, through extensive MATLAB simulations, we obtain the HOCs for different UAV heights, antenna element numbers, measurement gap values, and TTT values. 
    
    \item Next, we consider different distributions to approximate the HOC PMFs and show that the Gaussian distribution shows the best match. Then, by using the MATLAB \textit{curve fitting toolbox}, we express the Gaussian parameters with respect to different UAV speeds, GBS densities, handover parameters, UAV heights, and antenna configurations. Our results show that the HOCs trends have low variability with respect to GBS densities and UAV makes fewer handovers for high UAV speed and large TTT.

    \item Using the approximated HOC PMF, an expression for the Cramer-Rao lower bound (CRLB) of the estimated speed is derived. Moreover, a minimum variance unbiased (MVU) speed estimator analysis is provided, and we also propose a simple biased estimator. We show that for low TTT values, this biased estimator becomes an MVU estimator. For higher GBS densities, the variances of both the estimator matches closely with the CRLB, and we investigate the accuracy of this estimator for various UAV velocities, GBS densities, HOC measurement time intervals, TTTs, and UAV heights.
    
    \item The estimated speed is used to predict the mobility state (low/medium/high) of the UAV. We also derive the expressions of the probability of detection and probability of false alarm for the mobility states. 
    
    \item Finally, we study the proposed speed estimator and MSD techniques for a variable UAV speed and different HOC measurement time intervals and handover parameters.  
    
\end{itemize}

\begin{figure}[t]
\centering{\includegraphics[width=0.75\linewidth]{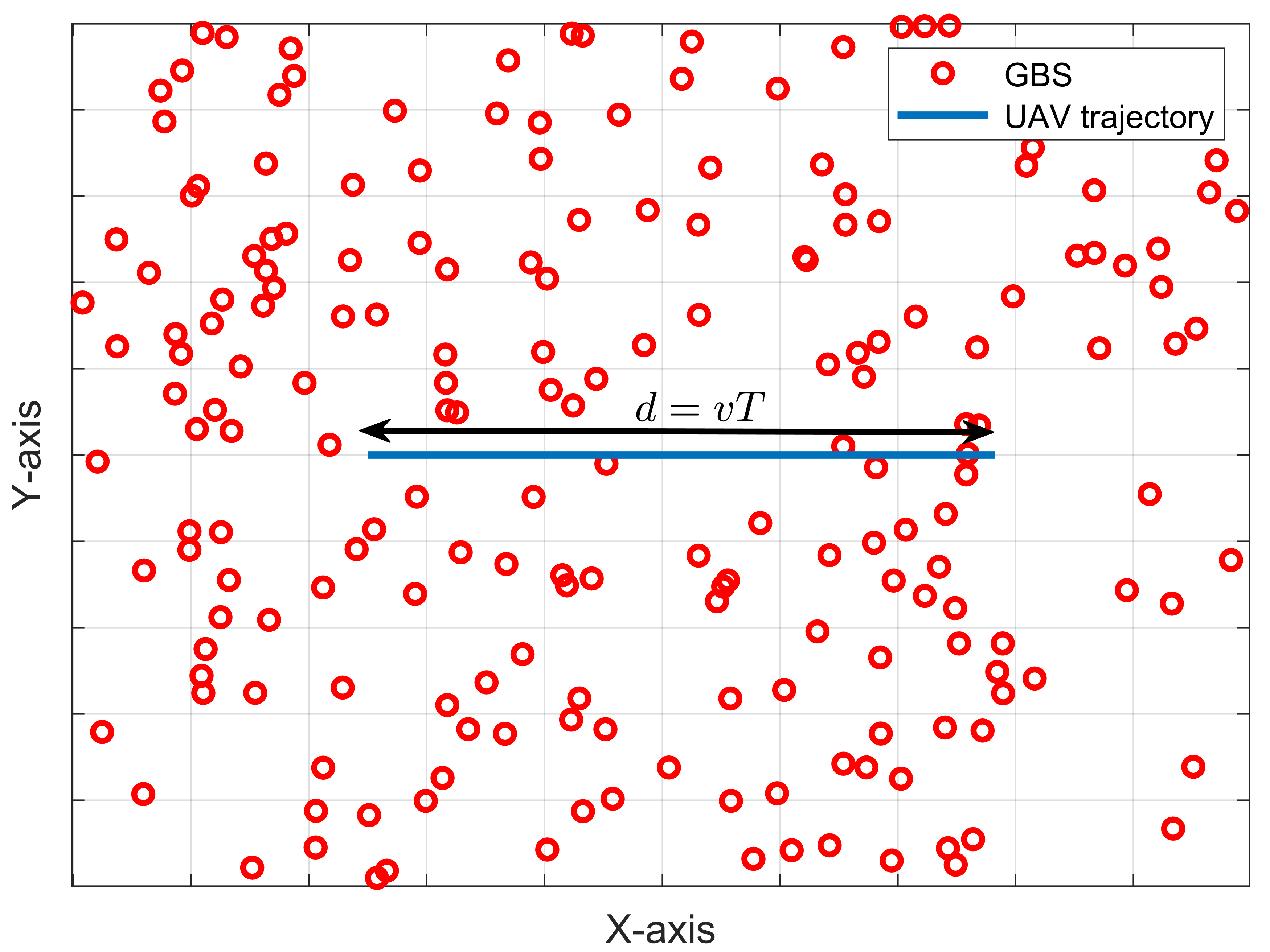}}
    \caption{Illustration of cellular network with linear mobility model.}
    \label{fig:uav_linear_trajec}
\end{figure} 

The rest of this paper is organized as follows. In Section~\ref{sec:lit_review}, we provide a literature review related to UAV speed estimation. Section~\ref{sec:sys} describes the system model for handover PMF calculation. The approximation of the HOC PMF using the Gaussian distribution is presented in Section~\ref{sec:stat}. We derive the CRLB for UAV speed estimation and provide a biased speed estimator in Section~\ref{sec:crlb}. In Section~\ref{sec:MSD analysis}, we present the expression of the mobility state (low, medium, or high) probabilities, the probability of detection, and the probability of false alarm. Simulation results are presented in Section \ref{sec:simulation}. Finally, Section~\ref{sec:Conc} concludes this paper. 
\section{Related works}
\label{sec:lit_review}
Existing cellular networks can estimate the mobility state of a UE into three classes: low, medium, and high mobility~\cite{david2012}. There are several works in the literature regarding the speed or mobility state estimation of ground UE. For instance, in \cite{arvind_ho}, the authors presented approximate probability mass functions (PMFs) for HOC, and based on them proposed an efficient estimator for ground UE speed. Using tools from stochastic geometry, analytical studies for handover rate in typical cellular networks are conducted in~\cite{lin_mobility}, while authors in~\cite{bao_mobility} also considered the presence of small base stations along with GBSs. Based on simple geometry, analytical analysis of handover failure performance in a dense heterogeneous network in the presence of a Rayleigh fading channel is studied in \cite{karthik2017}. The authors in~\cite{HO_oaram_opt} introduces a novel handover performance optimization algorithm by tuning the values of the handover related parameters in an automated manner. Their algorithm provides different weights to different handover performance metrics. A reinforcement learning (RL) based offline handover optimization scheme for ground UE in a 5G cellular network was proposed in \cite{5g_rl}. However, none of these prior works took the mobility management of aerial UE or UAVs into account.\looseness=-1 

In~\cite{amin_HO}, authors proposed a machine learning-based handover and resource management algorithm for aerial UE while considering the interference and HOCs of ground UE in the uplink (UL) scenario. Authors in \cite{challita_2019}, study the interference-aware optimal path planning of cellular-connected UAVs using RL in the UL scenario. But they do not consider the handover or mobility constraints of the UAVs. In~\cite{Ramy_walid_ICC2020}, the authors explored the effects of practical antenna configurations on the mobility management of cellular-connected UAVs. In a recent work, the authors explore the RL algorithm to maximize the received signal quality at a cellular-connected UAV while minimizing the number of handovers~\cite{xingqin_rl_2019}. In their proposed framework, the predefined trajectory is divided into discrete states and at each state, the UAV can decide which GBS to choose to maximize the received signal quality while minimizing the HOC. Motivated by this, the authors in~\cite{moin_ICC} also used RL to provide better connectivity to cellular-connected UAVs by tuning the GBS antenna configurations. In~\cite{ttt_uav}, the authors study the handover failure and ping-pong probabilities of a ground UE connected to UAV base stations while considering the handover parameters. In~\cite{3dbeam_HO}, the authors studied the performance of 3D beamforming for $5$G-connected UAVs assuming perfect beam alignment.

Real-world experiments were also conducted to test the feasibility of integrating UAVs as UE in \cite{lin_sky,lin_field,denmark_uav_test}. For example, the Third Generation Partnership Project (3GPP) also studied the challenges in providing reliable UAV mobility support in~\cite{3gpp}. By studying the performance of a cellular-connected UAV network in terms of RLF and the number of handovers, the authors conclude that the existing cellular networks will be able to support a small number of aerial UEs with good mobility support \cite{xingqin_ho_2018}. However, none of these prior works considered the problem of MSD of cellular-connected UAVs by estimating the speed based on the HOCs. To the best of the authors' knowledge, this is the first attempt to estimate both the speed and mobility state of a flying UAV in a realistic cellular network based on the HOC statistics. We compare our work with the state of the art in the literature in Table~\ref{tab:lit_review}.
\section{System Model}
\label{sec:sys}
\subsection{Network Model}
We consider a cellular network in which a single UAV (acting as an aerial UE), is flying along with a two dimensional (2D) linear trajectory (for instance, through the horizontal X-axis) at a fixed height $h_{\textrm{UAV}}$ and speed $v$. We consider the linear mobility model due to its simplicity and suitability for UAVs flying in the sky with virtually no obstacle e.g., in UAV corridors. The underlying cellular network consists of GBSs that are deployed with homogeneous Poisson point process (HPPP) $\Phi$ of intensity $\lambda_{\text{GBS}}~ \text{GBSs/km}^2$~\cite{Ramy_walid_ICC2020}, and all GBSs have similar height $h_{\textrm{GBS}}$ and transmission power $P_{\textrm{GBS}}$. In Fig.~\ref{fig:uav_linear_trajec}, we provide an illustrative example of the linear UAV mobility model in an area of $10\times10$ $\text{km}^2$ with $\lambda_{\text{GBS}}= 2~\text{GBSs/km}^2$. While flying, we assume that the network can track the number of handovers $H$ made by the UAV during a measurement time window $T$. We denote the distance traveled during this measurement duration as $d=vT$. We present the handover procedure later in this Section.

The GBSs consist of $N_t$ vertically placed cross-polarized directional antennas down-tilted by angle $\phi_{\rm d}$~\cite{Ramy_walid_ICC2020,Moin-3d}. We assume that the UAV is equipped with an omnidirectional antenna and the UAV is capable of mitigating the Doppler effect at its end~\cite{rui1}. We consider the GBS antennas to be omnidirectional in the horizontal plane but they have a variable radiation patterns along the vertical dimension with respect to the elevation angle between the antennas and the UE~\cite{ramy_antenna}. The $N_t$ antennas are equally spaced where adjacent elements are separated by half-wavelength distance. The element power gain (in dB) in the vertical plane at elevation angle $\theta$ with respect to the down-tilted antennas can be specified by~\cite{3gpp.38.901}
\begin{equation}
    G_e(\theta)=G_e^{\textrm{max}}- \text{min}\left\{ 12\left(\frac{\theta}{\theta_{3\mathrm{dB}}} \right)^2, \mathrm{G_m}\right\},
\end{equation}
where $\theta \in [-90^\circ, 90^\circ]$, $\theta_{3\textrm{dB}}$ refers to the $3$ dB beamwidth with a value of $65^\circ$, $G_e^{\textrm{max}}=8$~dBi is the maximum gain of each antenna element, and $\mathrm{G_m}$ is the side-lobe level limit, respectively, with a value $30$ dB~\cite{chowdhurySPAWC2020}. Note that the elevation angle $\theta=0^\circ$ refers to the 
horizon and the $\theta=90^\circ$ represents case when the main beam is facing upward perpendicular to the $xy$-plane~\cite{3gpp.38.901}. The array factor $A_f(\theta)$ of the ULA with $N_t$ elements while considering a down-tilt angle $\phi$ is given by
\begin{equation}
    A_f(\theta)=\frac{1}{\sqrt{N_t}}\frac{\sin\big({\frac{N_t\pi}{2}} (\sin\theta-\sin\phi)\big)}{\sin\big({\frac{\pi}{2}} (\sin\theta-\sin\phi)\big)}.
\end{equation} 
Let us denote $G_f(\theta)\triangleq 10\log_{10}( A_f(\theta))^2 $ as the array power gain in dB scale. Then the overall antenna gain at elevation angle $\theta$ is given by 
\begin{equation}
    G(\theta)=G_e(\theta)+G_f(\theta).
\label{eq:total_antenna_gain_down}
\end{equation}
\subsection{Ground Reflection Channel Model}
We consider a channel model that is characterized by both distance-based path-loss and ground reflection~\cite{chowdhury2021ensuring}. Let the length of the $3$D Cartesian distance from the UAV to a GBS $j$ be $l_j$ and the length of the incident and reflected paths are $r_{1,j}$ and $r_{2,j}$, respectively. According to this model, the received power from GBS $j$ at a UAV at height $h_{\rm UAV}$ can be specified as~\cite{chowdhury2021ensuring}:
\begin{equation}
    P_j=P_{\textrm{GBS}}\bigg[\frac{\lambda}{4\pi}\bigg]^2\bigg| \frac{\hat{G}_j(\theta)}{l_j} + \frac{R(\psi_j)\widetilde{G}_j(h_{\rm UAV})e^{i\Delta \phi_j}}{r_{1,j}+r_{2,j}}\bigg|^{\alpha(h_{\rm UAV})},
\label{eq:rx_power_gr}
\end{equation}
where $\theta$ is the elevation angle with respect to the down-tilted antenna of GBS $j$, $i=\sqrt{-1}$ is the imaginary unit of a complex number, $\lambda$ is the wavelength of the carrier frequency ${\text{f}_c}$, $\hat{G}_j(\theta)$ and  $\widetilde{G}_j(h_{\rm UAV})$ represent the height-dependent antenna gain of the direct and reflected path, respectively, $R(\psi_j)$ is the ground reflection coefficient for the angle of reflection $\psi_j$ with respect to the ground plane, $\Delta \phi_j=(r_{1,j}+r_{2,j})-l_j$ is the phase difference between the reflected and the direct signal paths, and $\alpha(h_{\rm UAV})$ is the height dependent propagation coefficient for UAV height $h_{\rm UAV}$. It is worth noting that our analysis can also be extended for 3GPP specified path loss models for UAVs~\cite{3gpp} as done in our previous work in~\cite{chowdhurySPAWC2020}.
\begin{figure}[t]
\centering{\includegraphics[width=0.75\linewidth]{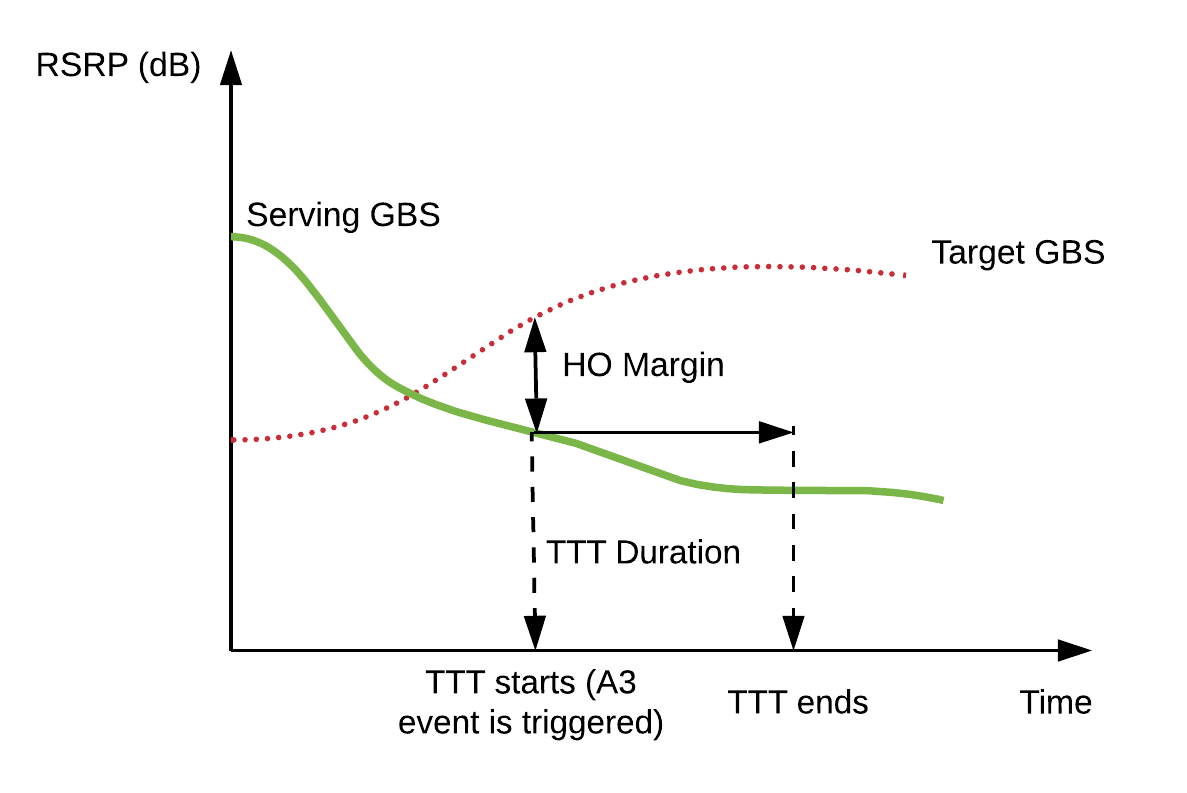}}
    \caption{An illustration of handover procedure which can get initiated at the end of the TTT duration by the serving GBS.}
    \label{fig:HO_procedure}
\end{figure}

Note that the ground reflection coefficient for cross-polarized antennas can be calculated as $R(\psi_j)=\frac{R_{\rm H}(\psi_j)-R_{\rm V}(\psi_j)}{2}$~\cite{najibi2013physical,wietfeld_ground_reflection},~which depends on the reflection coefficients for horizontal linear polarization $R_{\rm H}(\psi_j)$ and vertical linear polarization $R_{\rm V}(\psi_j)$. Moreover, $\hat{G}_j(\theta)$ depends on the instantaneous elevation angle between the GBS and the UAV by~\eqref{eq:total_antenna_gain_down}, whereas $\widetilde{G}_j(h_{\rm UAV})$ can be expressed as:
\begin{equation}
\label{eq:ground_reflected}
  \resizebox{\hsize}{!}{$\widetilde{G}_j(h_{\rm UAV}) = \left \{
  \begin{aligned}
    &\hat{G}_j(\psi_j), && h_{\rm UAV}<h_t \\
    &\frac{\hat{G}_j(\psi_j)}{2}, && h_t \leq h_{\rm UAV}\leq 2h_t\\
    &\frac{\hat{G}_j(\psi_j)}{2}-\frac{h_{\rm UAV}}{2 h_{t,c}}\cdot(\hat{G}_j(\psi_j)-1), && 2h_t \leq h_{\rm UAV}\leq 500 \\
    &0.5, && h_{\rm UAV}\geq 500 
  \end{aligned} \right.$}
\end{equation} 
where $h_t=2h_{\textrm{GBS}}+2$ and $h_{t,c}=500$~m are threshold heights~\cite{wietfeld_ground_reflection}, and $\hat{G}_j(\psi_j)$ is the antenna gain of the incident path on the ground from the down-tilted antennas which depends on $N_t$. Finally, the height-dependent propagation coefficient can be expressed as:
\begin{equation}
\label{eq:ground_refle}
  \alpha(h_{\rm UAV})= \left \{
  \begin{aligned}
    &\alpha_0-h_{\rm UAV}\cdot \bigg(\frac{\alpha_0-2}{h_{\textrm{GBS}}}\bigg) &&  h_{\rm UAV}<2\cdot h_{\textrm{GBS}}, \\
    &2 && h_{\rm UAV}\geq 2\cdot h_{\textrm{GBS}},
  \end{aligned} \right.
\end{equation} 
where $\alpha_0=3.5$ is the maximum possible attenuation coefficient~\cite{wietfeld_ground_reflection}. 
\subsection{Handover Procedure}
In a traditional cellular network, a UAV will measure the RSRPs from all the adjacent GBSs at subsequent measurement gaps using \eqref{eq:rx_power_gr}. Here, we consider a handover mechanism that involves a handover margin (HOM) parameter, and a TTT parameter, which is a time window that starts after meeting the following handover condition (A3 event \cite{3gpp.36.331}):
\begin{equation}
\text{RSRP}_{{j}}>\text{RSRP}_{{i}}+m_{\text{hyst}},
\label{a3_event}
\end{equation}
where $\text{RSRP}_{{j}}$ and $ \text{RSRP}_{{i}}$ are the RSRPs (in dB) measured from the serving GBS $i$ and target GBS $j$, respectively, and $m_{\text{hyst}}$ is the HOM set by the network operator. Throughout the flight duration, the UAV can measure the RSRPs from all the adjacent GBSs at subsequent measurement gaps. The UAV does not transmit its measurement report to its current serving GBS before the TTT expires~\cite{karthik2017}. An illustration of the handover mechanism is depicted in Fig.~\ref{fig:HO_procedure}. When the condition in \eqref{a3_event} is satisfied for the first time, the UAV waits for a duration of TTT, before sending a measurement report to its serving GBS to initiate the actual handover.

Note that the use of TTT and HOM are critical to ensure that successive handovers among neighboring GBSs (ping pong event) due to fluctuations in the link qualities from different GBSs are minimized. If the handover event entry condition is still satisfied after TTT, the UAV sends a measurement report to its associated GBS, which then communicates with the target GBS to perform the handover. Hereinafter, we will use $t_{\rm TTT}$ and $t_{\rm MG}$ to represent TTT and measurement gap, respectively.

\begin{figure}[t]
\centering
	\subfloat[]{
			\includegraphics[width=.48\linewidth]{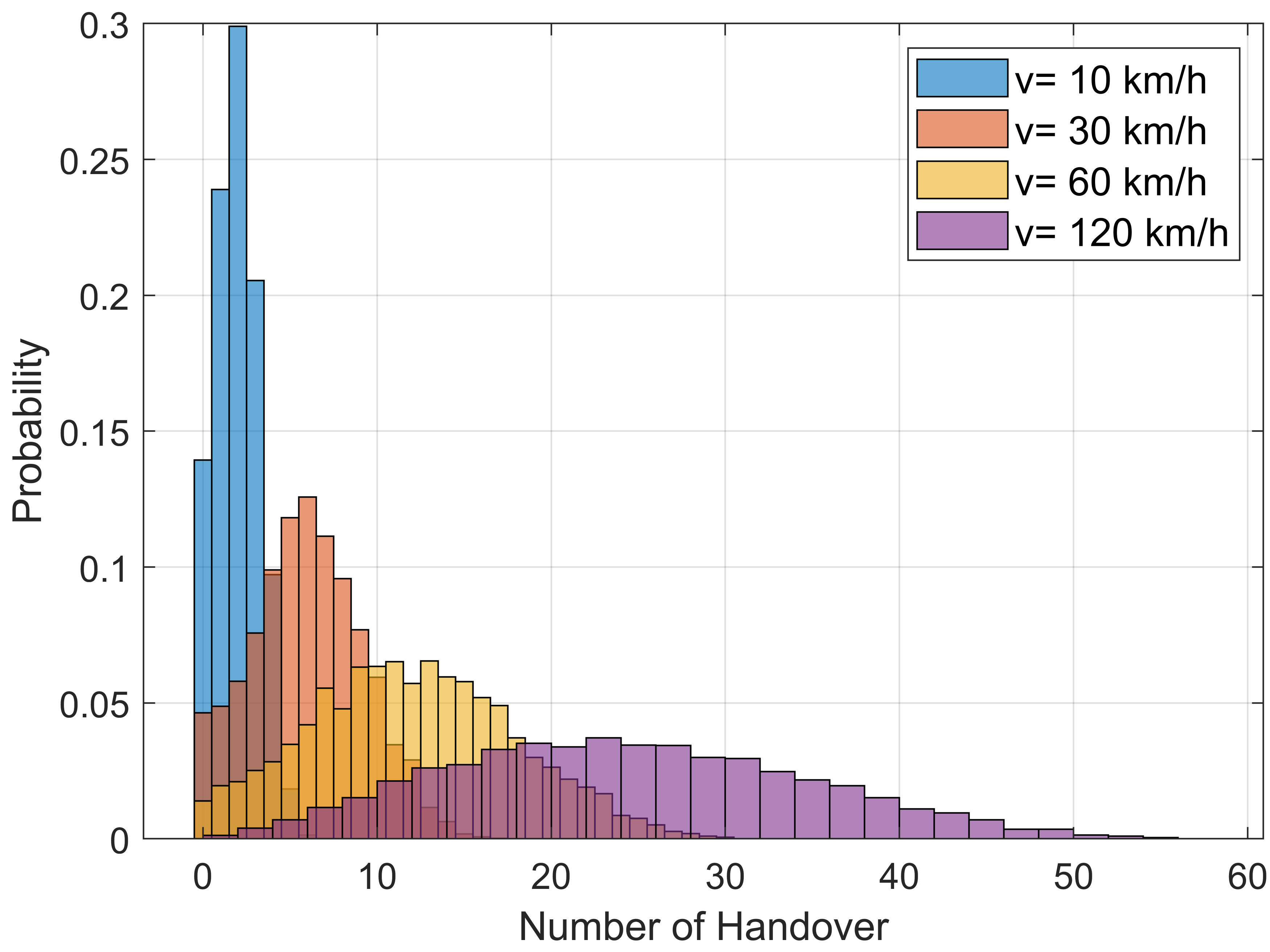}}
		\subfloat[]{
			\includegraphics[width=.48\linewidth]{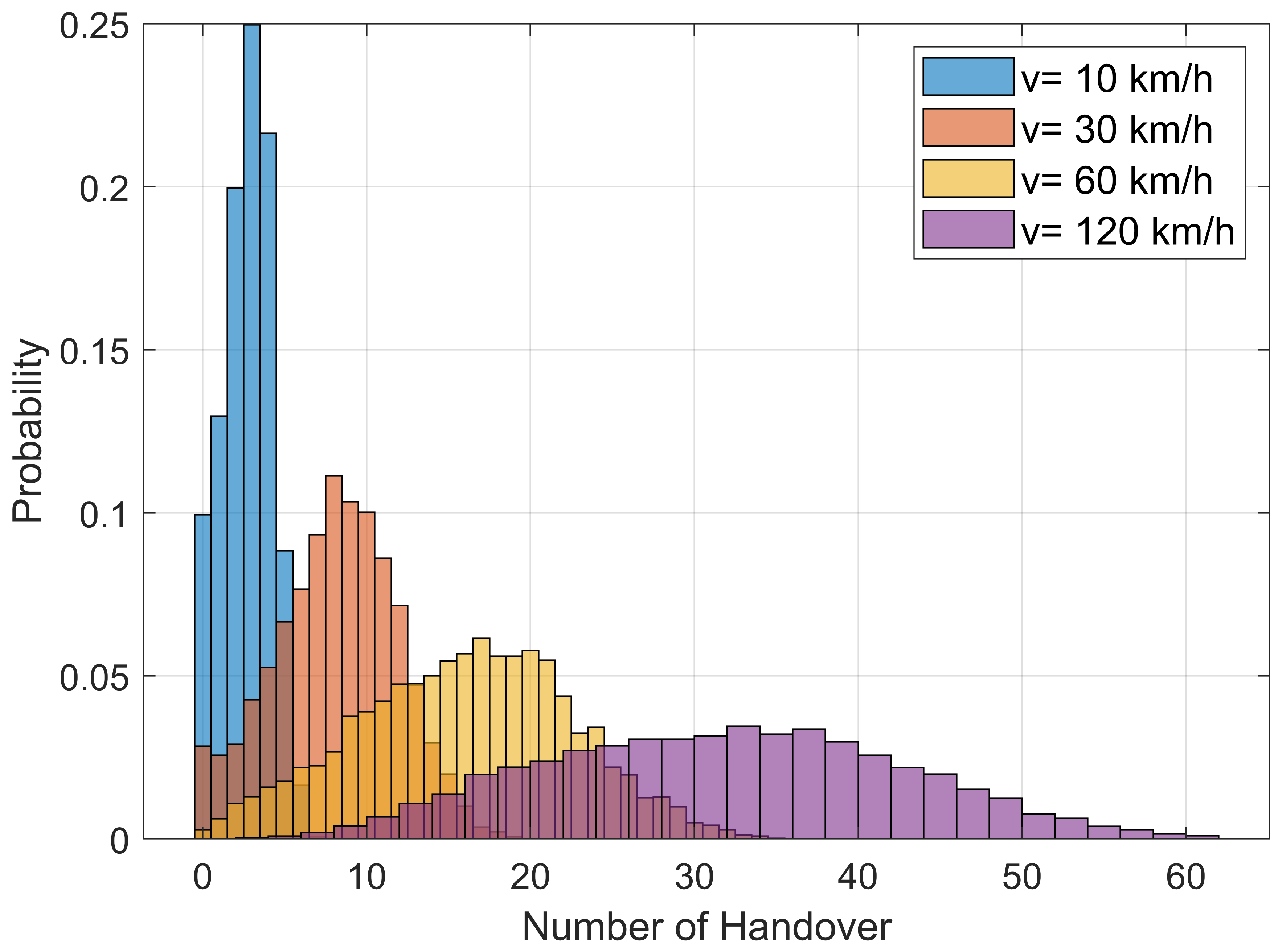}} 
			
    \caption{{PMF of the HOC for different $\lambda_{\text{GBS}}$ and $v$ considering ground reflection, $t_{\rm MG}=$$40$~ms, and $t_{\rm TTT}=$$0$~ms; (a) $\lambda_{\text{GBS}}= 1~\text{GBSs/km}^2$; (b) $\lambda_{\text{GBS}}= 3~\text{GBSs/km}^2$.}}
    \label{fig:uav_velocity_PMF_empritical}
\end{figure}

\begin{figure}[t]
\centering
	\subfloat[]{
			\includegraphics[width=.48\linewidth]{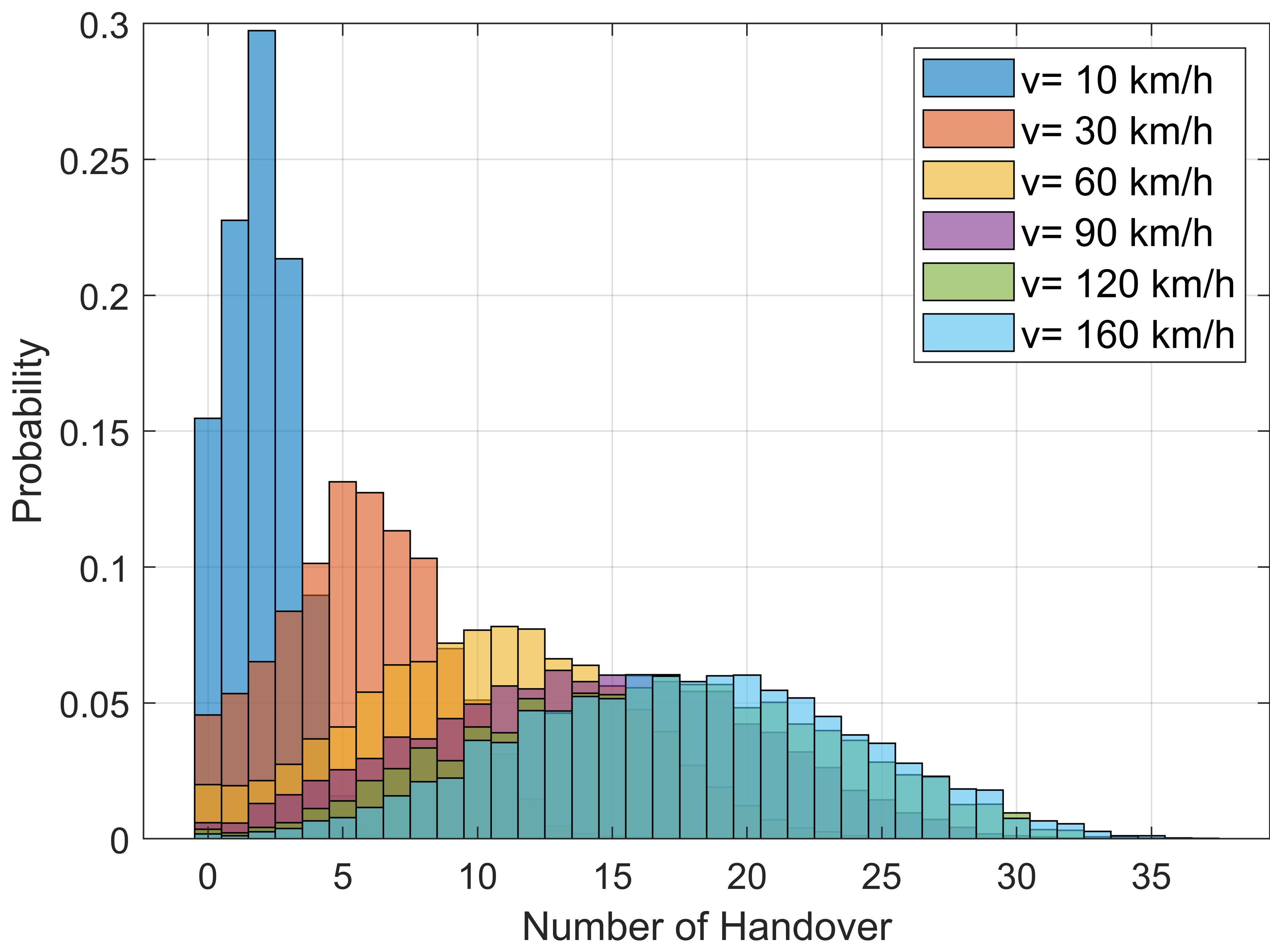}}
		\subfloat[]{
			\includegraphics[width=.48\linewidth]{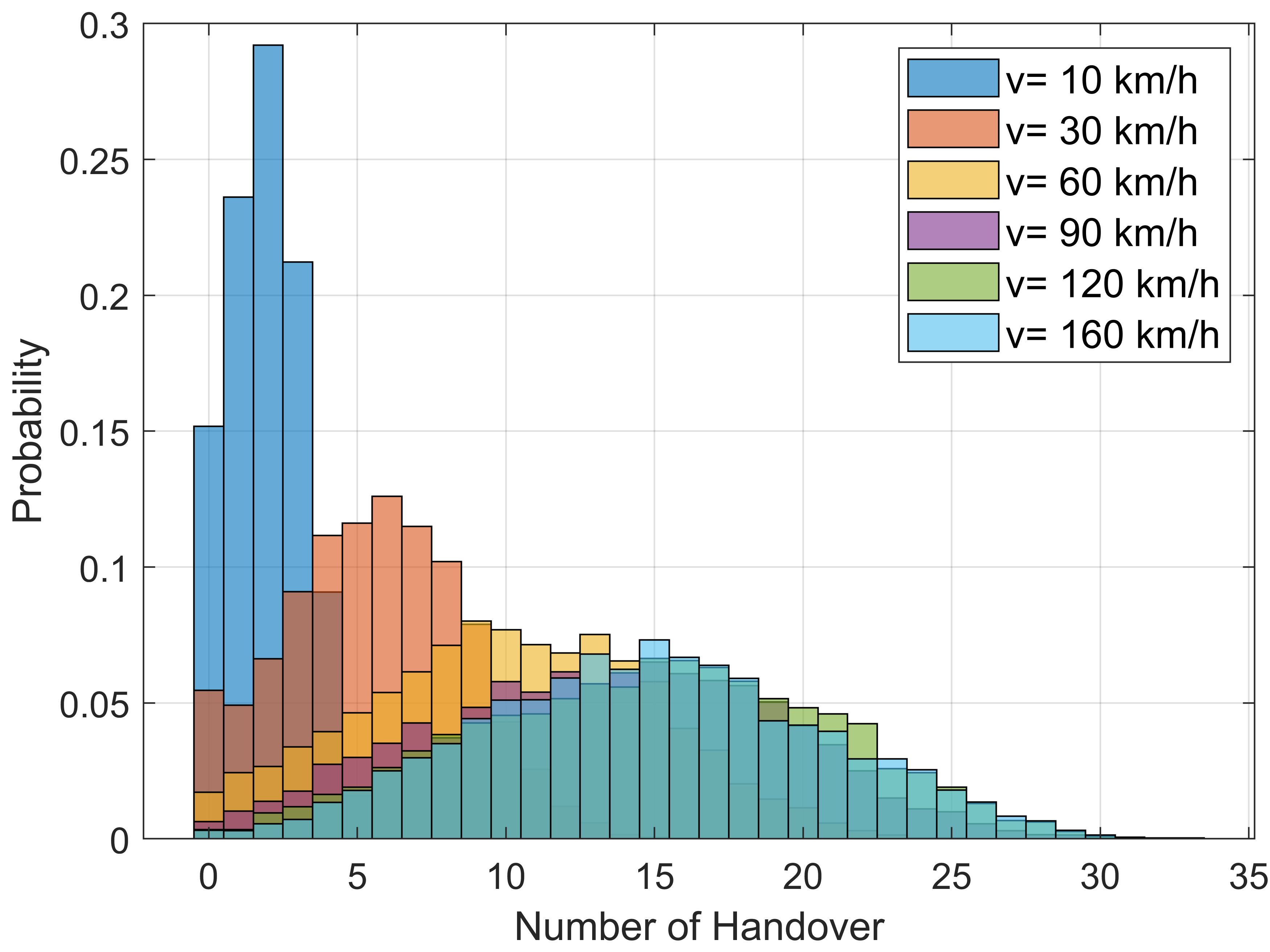}} 
			
    \caption{{PMF of the HOC for different $\lambda_{\text{GBS}}$ and $v$ considering ground reflection and $t_{\rm MG}=160$~ms; (a) $t_{\rm MG}=40$~ms; (b)  $t_{\rm MG}=100$~ms.}}
    \label{fig:velocity_fitting_160ms}
\end{figure}

\section{Approximation of the Handover-Count Statistics using Gaussian Distribution }
\label{sec:stat}

In this section, we introduce an approximation for the PMF of handover count using Monte Carlo simulations. The parameters of the Gaussian distribution will be approximated based on $\lambda_{\text{GBS}}$ and UAV speed $v$ using the curve fitting tools
in MATLAB. 

\subsection{Approximation of the PMF of Handover Count Using Gaussian Distribution}

For obtaining the estimated speed of a UAV based on its HOC, we need to know the HOC PMF $f_H(h)$. To the best of our knowledge, there exists no closed-form expression for the PMF of HOC of a cellular-connected UAV. Moreover, due to the intractability of the GBS antenna radiation pattern, handover process, and ground reflection based channel model, it is extremely difficult to obtain an exact expression of $f_H(h)$.

In Fig.~\ref{fig:uav_velocity_PMF_empritical}, we plot the HOC PMFs from extensive MATLAB simulations for various $v$ and $\lambda_{\text{GBS}}$ values. For each combination of $v$ and $\lambda_{\text{GBS}}$, we obtained samples of HOC $H$ for constructing the PMF $f_H(h)$. Here, we consider the HOC measurement time interval $T=12$~s~\cite{arvind_ho}. For low values of $\lambda_{\text{GBS}}$, as depicted in Fig.~\ref{fig:uav_velocity_PMF_empritical}(a), the PMFs for different UAV velocities overlap with each other significantly. For higher values of $\lambda_{\text{GBS}}$, the PMFs still overlap, but they are slightly more spread out which will lead to more accurate speed estimation. One interesting observation is that the overall trends of the PMFs do not change significantly with higher GBS densities. This happens because the coverage in the sky is  \emph{fragmented} due to the weak sidelobes of the GBS antennas while serving UAVs and the ground reflection effect from the down-tilted main lobes, see e.g. the aerial cellular coverage results presented in~\cite{lin_field,chowdhury2021ensuring}. The nature of the fragmented coverage in the sky
has a low correlation with the number of GBSs on the ground. In other words, the fragmented coverage pattern will still be present even for a lower GBS density due to the antenna sidelobe and the ground reflection effects.

We also study the impact of the TTT value of $160$~ms on the HOC PMFs in Fig.~\ref{fig:velocity_fitting_160ms}. Interestingly, for $t_{\rm TTT}$~$=160$~ms, the PMFs of the higher values of UAV speed such as $120$~km/h and $160$~km/h, the corresponding PMFs overlap significantly with $v=90$~km/h. This is because the higher speed allows the UAV to travel over small coverage areas without making any handover. According to~\cite{karthik2017}, the handover failure probability for high-speed UE decreases with decreasing $t_{\rm MG}$. Hence, unless otherwise stated, we consider $t_{\rm MG}$ to be $40$~ms in this work. 

To obtain the best distribution fitting for the HOC PMF, we consider different distributions and plot the respective results for two UAV velocities in Fig.~\ref{fig:velocity_fitting_distributions}. From the obtained HOC data samples, we can conclude that the PMFs for different $v$ resemble closely with the Gaussian distribution, especially for higher speeds. The probability density function (PDF) of Gaussian distribution can be expressed as a function of mean parameter $\mu$ and variance parameter $\sigma^2$ as:
\begin{equation}
f^{(n)}(x)= \frac{1}{\sqrt{2\pi\sigma^2}}e^{-\frac{(x-\mu)^2}{2\sigma^2}}.   
\label{eq:normal}
\end{equation}

\begin{figure}[t]
\centering
	\subfloat[$v=10$ km/h]{
			\includegraphics[width=.75\linewidth]{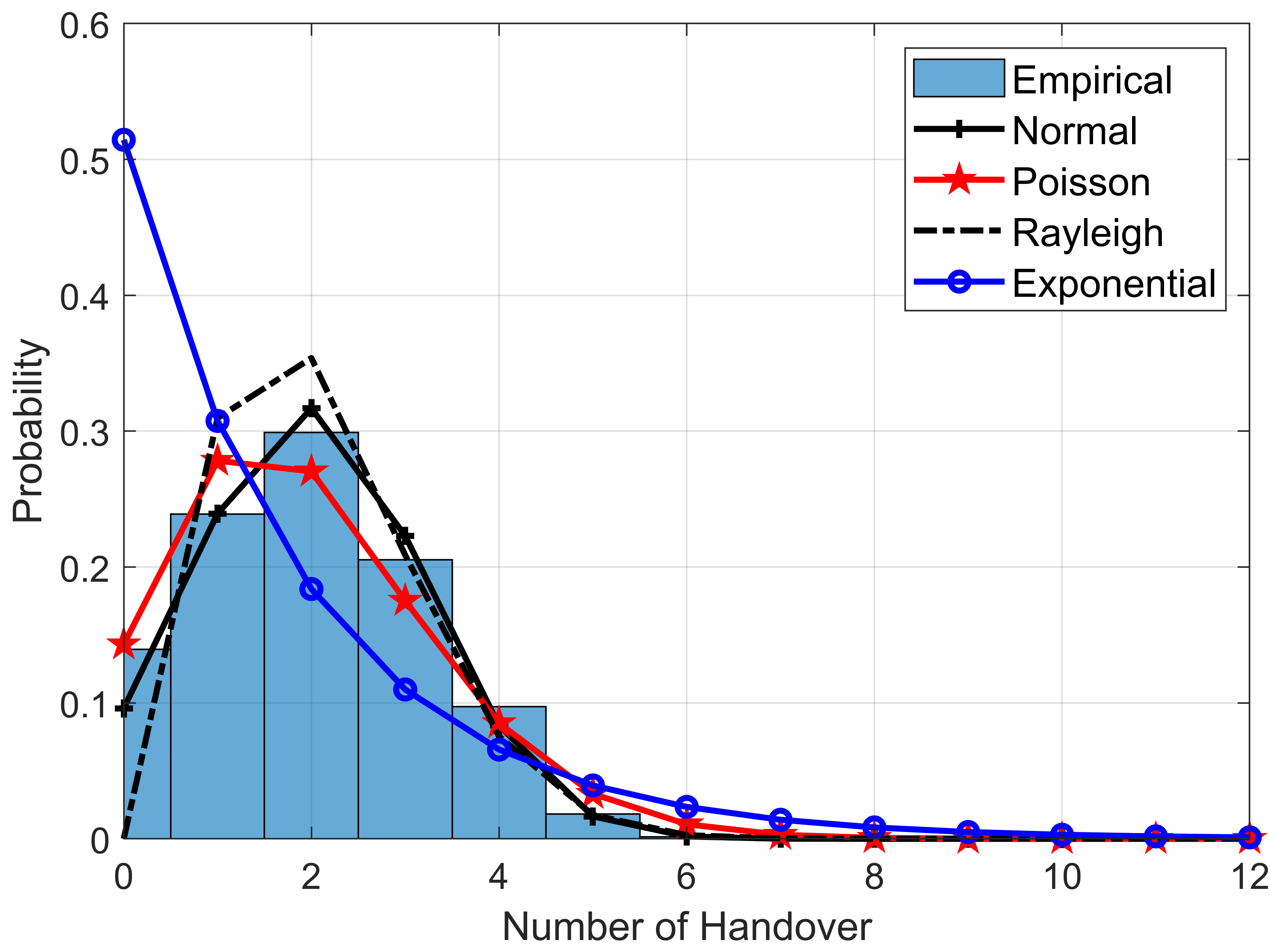}}\\
		\subfloat[$v=120$ km/h]{
			\includegraphics[width=.75\linewidth]{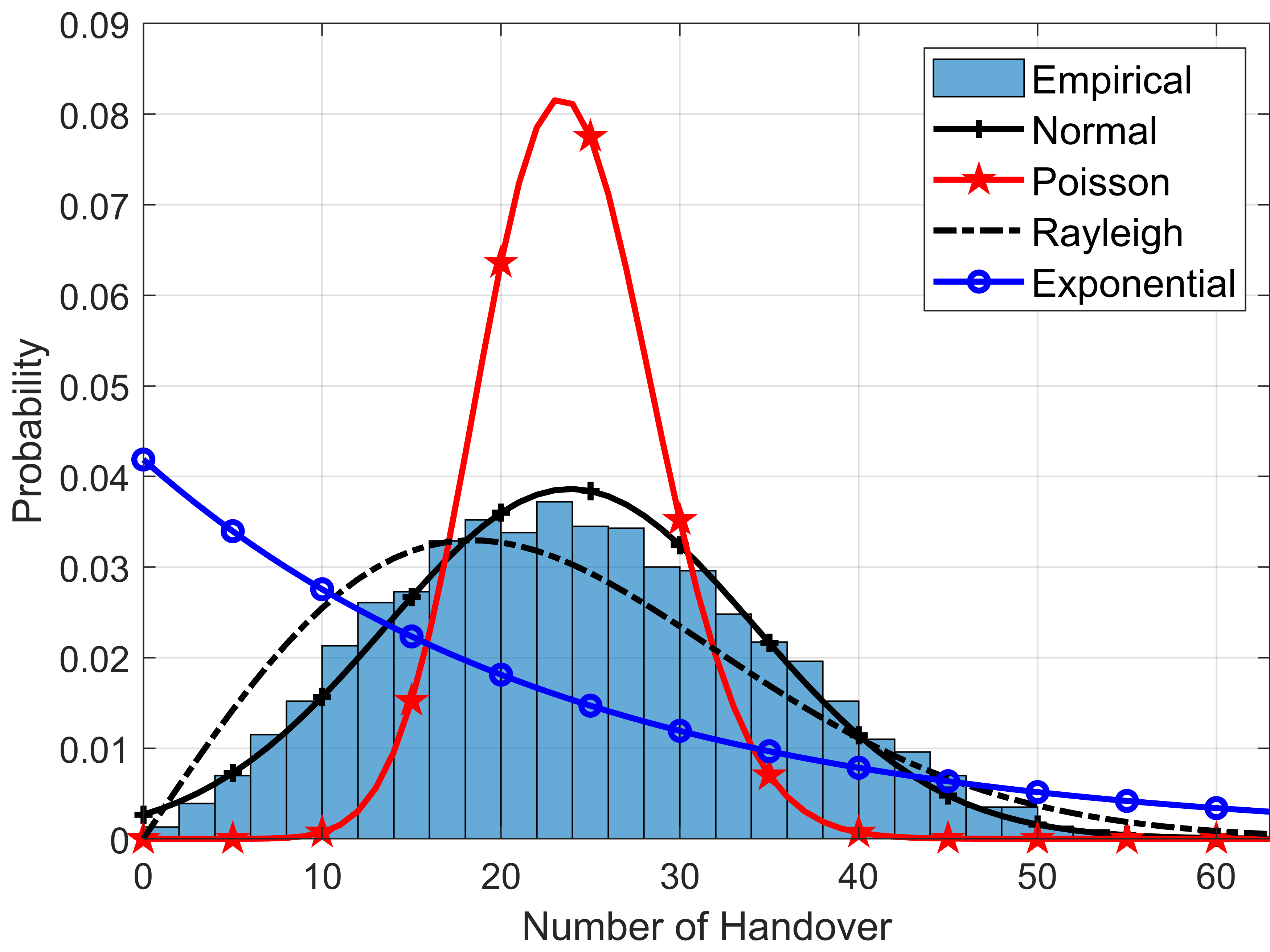}} 
    \caption{{PMF fitting of HOC for different distributions with  $\lambda_{\text{GBS}}=1$ and UAV speed $v$.}}
    \label{fig:velocity_fitting_distributions}
\end{figure}

Since the PMF $f_H(h)$ is discrete whereas the Gaussian distribution in \eqref{eq:normal} is continuous, we consider only non-negative integer samples of the Gaussian distribution for the fitting process~\cite{arvind_ho}. Hence, we approximate the PMF of HOC of a UAV flying with a constant speed with mean parameter $\mu$ and variance parameter $\sigma^2$ as,
\begin{equation}
f^{(n)}_H(h)= \frac{1}{\sqrt{2\pi\sigma^2}}e^{-\frac{(h-\mu)^2}{2\sigma^2}},~\text{for}~h \in \{0,1,2,...\}.   
\label{eq:PMF_normal}
\end{equation}
The values of the parameters $\mu$ and $\sigma^2$ should be chosen such that the mean square error (MSE) between $f_H(h)$ and $f_H^{(n)}(h)$ is minimized. For obtaining analytical expression, we express  $\mu$ and $\sigma^2$ as functions of distance $d=vT$ and $\lambda_{\text{GBS}}$. 
In what follows, for a given UAV height, antenna orientation, and down-tilt angle, we derive the heuristic closed-form expressions for the Gaussian parameters as functions of GBS density and UAV speed.
\begin{prop}
The $\mu$ and $\sigma^2$ parameters for the approximation in \eqref{eq:PMF_normal} that minimizes the MSE between $f_H(h)$ and $f_H^{(n)}(h)$ can be expressed as
\begin{equation}
\begin{aligned}
\mu&=a_1\times\lambda_{\text{GBS}}^{b_1} \times d^{c_1},\\
\sigma^2&=a_2\times\lambda_{\text{GBS}}^{b_2}\times d^{c_2},
\label{eq:PMF_param_var_normal}
\end{aligned}
\end{equation}
where $d=vT$ represents the distance traveled by the UAV during HOC measurement time $T$. The value of the three parameters involved for each of \emph{$\mu$} and \emph{$\sigma^2$} depends on the $t_{\rm MG}$, $t_{\rm TTT}$, $h_{\rm UAV}$, and antenna configurations.
\end{prop}

\begin{table}[t]
\begin{center}
\caption{Gaussian parameters for different $t_{\rm TTT}$ and $t_{\rm MG}$. }
\label{tab:mu_sigma_MG_TTT}
\subfloat[$\mu$ for different $t_{\rm TTT}$ and $t_{\rm MG}$]
{
\scalebox{0.75}
{
\begin{tabular}{cc|ccc}
\hline
\multicolumn{2}{c|}{\multirow{2}{*}{($a_1$,$b_1$,$c_1$)}} & \multicolumn{3}{c}{\textbf{$t_{\rm TTT}$ (ms)}}                                                                        \\ \cline{3-5} 
\multicolumn{2}{c|}{}                                     & \multicolumn{1}{c|}{0}                 & \multicolumn{1}{c|}{40}             & 160                   \\ \hline \hline
\multicolumn{1}{c|}{\multirow{2}{*}{\textbf{$t_{\rm MG}$ (ms)}}}    & 40     & \multicolumn{1}{c|}{($58.6,0.3048,1$)} & \multicolumn{1}{c|}{($55.2,0.29,1)$} & ($32.68,0.2221,0.746$) \\ \cline{2-5} 
\multicolumn{1}{c|}{}                            & 100    & \multicolumn{1}{c|}{($55.66,0.3013,1$)}                  & \multicolumn{1}{c|}{($52.91,0.29,1)$}               &   ($28.66,0.2268,0.6913$)                     \\ \hline
\end{tabular}
}
}

\subfloat[$\sigma^2$ for different $t_{\rm TTT}$ and $t_{\rm MG}$.]
{
\scalebox{0.75}
{
\begin{tabular}{cc|ccc}
\hline

\multicolumn{2}{c|}{\multirow{2}{*}{($a_2$,$b_2$,$c_2$)}} & \multicolumn{3}{c}{\textbf{$t_{\rm TTT}$ (ms)}}                                                                        \\ \cline{3-5} 
\multicolumn{2}{c|}{}                                     & \multicolumn{1}{c|}{0}                 & \multicolumn{1}{c|}{40}             & 160                   \\ \hline \hline
\multicolumn{1}{c|}{\multirow{2}{*}{\textbf{$t_{\rm MG}$ (ms)}}}    & 40     & \multicolumn{1}{c|}{($425.2,0.167,1.55$)} & \multicolumn{1}{c|}{($327.7,0.15,1.435)$} & ($109,0.032,0.963$) \\ \cline{2-5} 
\multicolumn{1}{c|}{}                            & 100    & \multicolumn{1}{c|}{($315.1,0.1212,1.4$)}                  & \multicolumn{1}{c|}{($263.8,0.1224,1.41)$}               &   ($81.93,0.018,0.8344$)                     \\ \hline
\end{tabular}

}
}
\end{center}
\end{table}
\begin{table}[t]
\begin{center}
\caption{Gaussian parameters for different $h_{\rm UAV}$ and $N_t$. }
\label{tab:mu_sigma_height_antenna}
\subfloat[$\mu$ and $\sigma^2$ for different $h_{\rm UAV}$.]
{
\scalebox{0.75}
{
\begin{tabular}{cl|ccc}
\hline
\multicolumn{2}{c|}{\multirow{2}{*}{}} & \multicolumn{3}{c}{\textbf{$h_{\rm UAV}$} (m)}                                                     \\ \cline{3-5} 
\multicolumn{2}{c|}{}        & \multicolumn{1}{c|}{80}                 & \multicolumn{1}{c|}{100}                & 120                   \\ \hline \hline
\multicolumn{2}{c|}{$(a_1,b_1,c_1)$}  & \multicolumn{1}{c|}{($72.65,0.2736,1$)} & \multicolumn{1}{c|}{($58.6,0.3048,1)$}  & ($37.09,0.4058,1$)    \\ \hline
\multicolumn{2}{c|}{($a_2,b_2,c_2$)}  & \multicolumn{1}{c|}{($593,0.1,1.539$)}  & \multicolumn{1}{c|}{$425.2,0.167,1.55$} & ($233.7,0.222,1.417$) \\ \hline

\end{tabular}}
}

\subfloat[$\mu$ and $\sigma^2$ for different $N_t$.]
{
\scalebox{0.75}
{
\begin{tabular}{cl|ccc}
\hline
\multicolumn{2}{c|}{\multirow{2}{*}{}} & \multicolumn{3}{c}{\textbf{Number of antenna elements~$N_t$}}                                                     \\ \cline{3-5} 
\multicolumn{2}{c|}{}        & \multicolumn{1}{c|}{4}                 & \multicolumn{1}{c|}{8}                & 16                   \\ \hline \hline
\multicolumn{2}{c|}{$(a_1,b_1,c_1)$}  & \multicolumn{1}{c|}{($144.50,0.28,1.00$)} & \multicolumn{1}{c|}{($58.6,0.3048,1)$}  & ($23.50,0.66,1$)    \\ \hline
\multicolumn{2}{c|}{($a_2,b_2,c_2$)}  & \multicolumn{1}{c|}{($2747,-0.23,1.409$)}  & \multicolumn{1}{c|}{$425.2,0.167,1.55$} & ($163.2,0.31,1.50$) \\ \hline
\end{tabular}
}}
\end{center}
\end{table}

\subsection{Characteristics of HOC Statistics for Gaussian Scenario}
To obtain simplified heuristic closed-form expressions of the Gaussian parameters, we tried with different linear and non-linear functions containing the variables $d=vT$ and $\lambda_{\text{GBS}}$ using the MATLAB curve fitting toolbox. We then obtain the 2D power fits where the values of $\mu$ and $\sigma^2$ can be obtained as in \eqref{eq:PMF_param_var_normal}. For $t_{\rm MG}$~$=40$~ms and $t_{\rm TTT}$~$=0$~ms, we report the values of the mean parameters as, $a_1=58.6$, $b_1=0.305$, and $c_1=1$. The values of the variance parameters are $a_2=425.2$, $b_2=0.167$, and $c_2=1.55$. The proposed expression of $\mu$ provided us a fitting with root mean squared error (RMSE) of $0.63$ and adjusted R-square value $0.9956$ for $t_{\rm MG}$~$=40$~ms and $t_{\rm TTT}$~$=0$~ms. 

In Fig.~\ref{fig:ho_count_fit_normal}, we show the trend of these two parameters with respect to $d$ and $\lambda_{\text{GBS}}$. The accuracy of this approximation is justified in Fig.~\ref{fig:mu_lambda_theo_sim} where the expression is plotted in comparison with the plots obtained through simulations. From Fig.~\ref{fig:mu_lambda_theo_sim}, we can conclude that the Gaussian parameters do not show substantial variability with respect to the GBS density. It is worth noting that, we consider $v$ up to $120$ km/h for the fitting purpose. However, this approximation also holds well for $v=160$ km/h considering $t_{\rm TTT}$~$=0$~ms, which will be shown in Section~\ref{sec:simulation}.

We also run extensive simulations for different $h_{\rm UAV}$, $t_{\rm TTT}$, $t_{\rm MG}$, and $N_t$ to obtain the respective Gaussian PMF parameters. Table~\ref{tab:mu_sigma_MG_TTT}(a) and Table~\ref{tab:mu_sigma_MG_TTT}(b) present the parameters involved with $\mu$ and $\sigma^2$, respectively for different measurement gaps and TTT values. Note that higher values of $\mu$ and $\sigma^2$ correspond to a higher number of handovers made by the UAV within measurement duration $T$. Observing the values of these parameters, we can conclude that the UAV makes fewer handovers for higher TTT values since it skips some handover triggering events. For a similar reason, larger $t_{\rm MG}$ slightly decrease the HOC. Another interesting observation is that the values of $\sigma^2$ are hardly dependent on the GBS density for higher $t_{\rm TTT}$. 

In  Table~\ref{tab:mu_sigma_height_antenna}(a), we report the parameters for three different UAV heights. The UAV tends to make fewer handovers at higher altitudes due to low relative mobility which also corroborates the findings of~\cite{3dbeam_HO}. Finally, we run curve fitting for different $N_t$ and provide the results in Table~\ref{tab:mu_sigma_height_antenna}(b). The parameter values show that the UAV will make more handovers for lower $N_t$ values. This is because there exist stronger sidelobes for lower $N_t$ which contributes to more scattered coverage~\cite{Moin-3d, chowdhury2021ensuring}. Higher $N_t$ results in increased but weaker and narrower sidelobes. The variance of HOC PMF increases with the GBS density but overall decreases with $N_t$ due to lower values of $a_2$.    

\begin{figure}[t]
\centering
	\subfloat[]{
			\includegraphics[width=.99\linewidth]{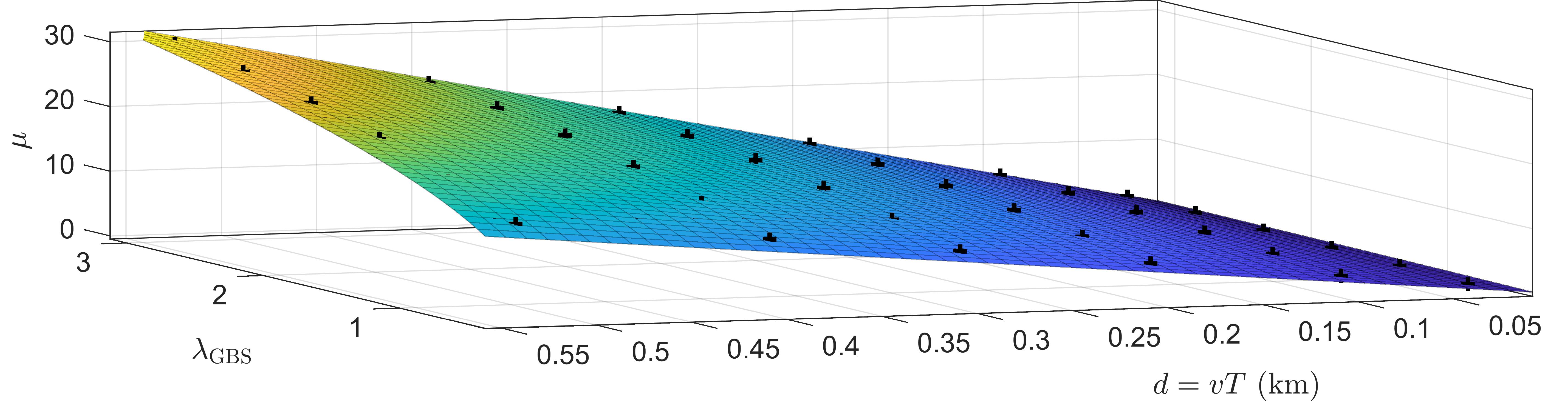}} 
			\hfill
		\subfloat[]{
			\includegraphics[width=.99\linewidth]{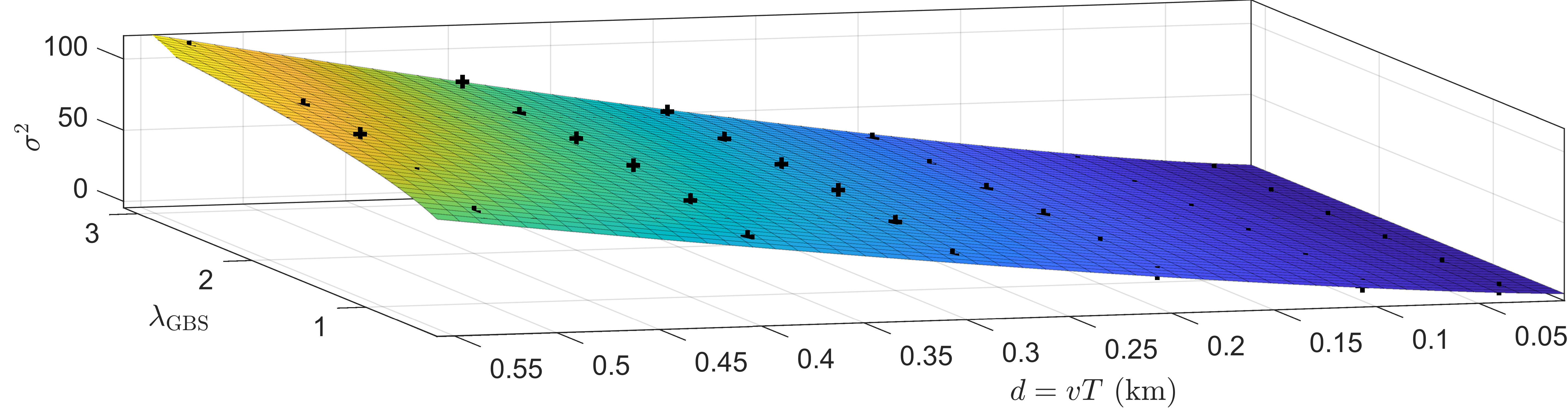}} 
    \caption{{Fitting of the HOC PMF for Gaussian parameters (a) $\mu$ and (b) $\sigma^2$ with respect to the covered distance $d=vT$ and GBS density $\lambda_{\text{GBS}}$. Here, we consider $t_{\rm MG}$~$=40$~ms and $t_{\rm TTT}$~$=0$~ms}}
    \label{fig:ho_count_fit_normal}
\end{figure}

\begin{figure}[t]
\centering
	\subfloat[]{
			\includegraphics[width=.75\linewidth]{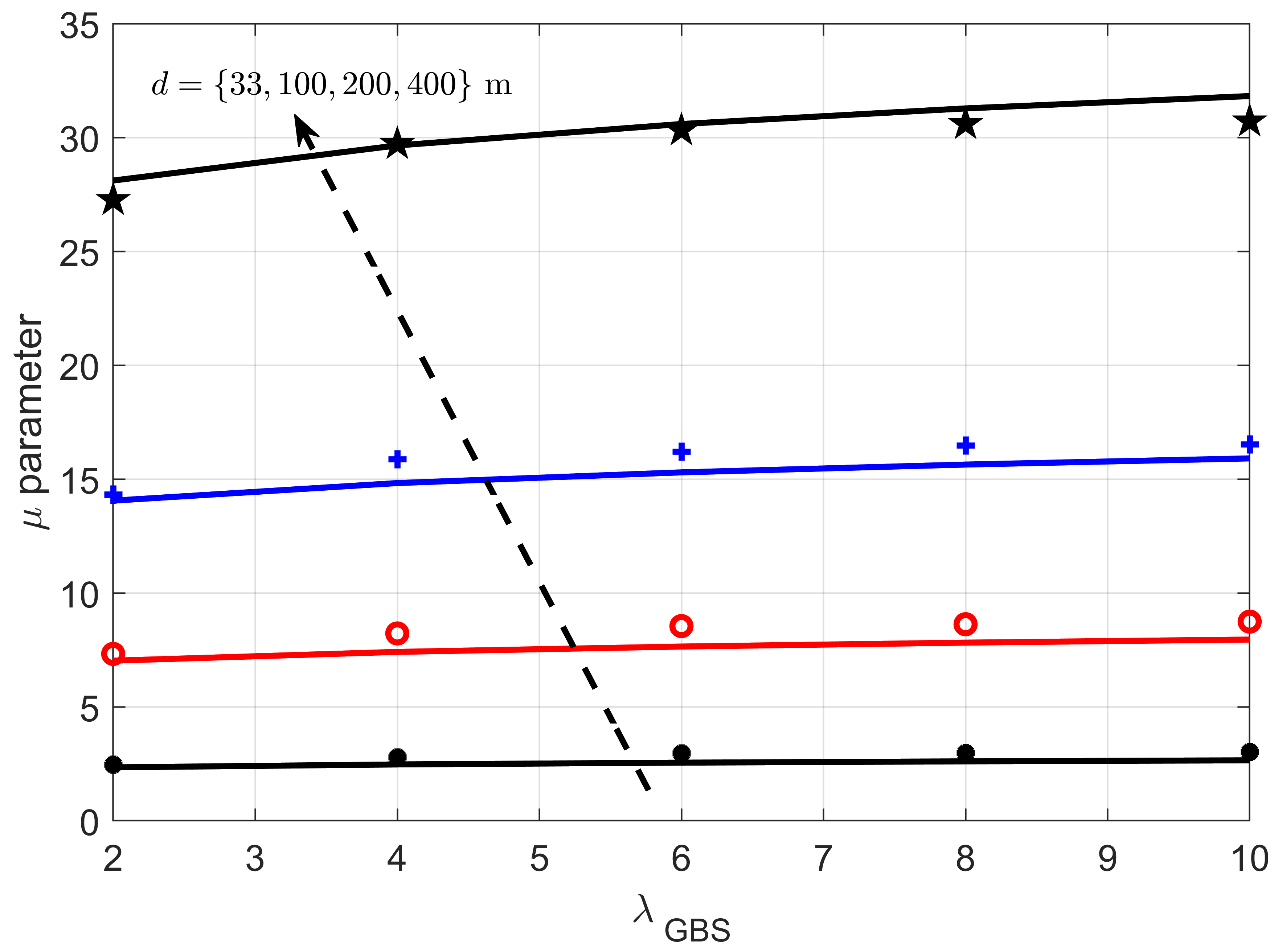}} 
			\hfill
		\subfloat[]{
			\includegraphics[width=.75\linewidth]{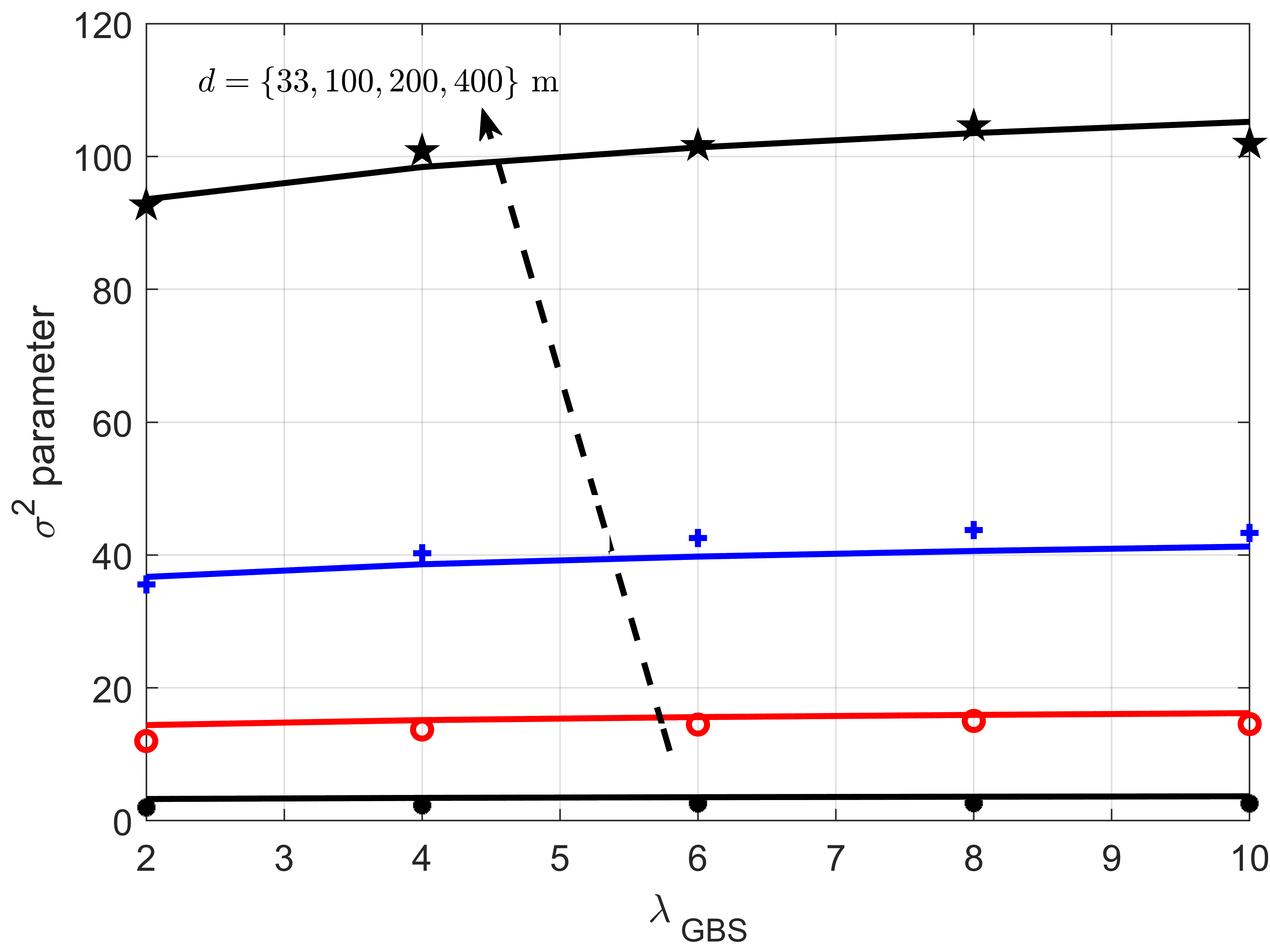}} 
    \caption{{Approximation of the Gaussian parameters (a) $\mu$ and (b) $\sigma^2$ for $t_{\rm MG}$~$=40$~ms and $t_{\rm TTT}$~$=0$~ms.}}
    \label{fig:mu_lambda_theo_sim}
\end{figure}

\section{Cramer-Rao Lower Bound for UAV speed Estimation}
\label{sec:crlb}
Cramer-Rao lower bound (CRLB) provides a lower bound on the variance of an unbiased estimator. An estimator is considered to be unbiased if the expected value of the estimates coincides with the true value of the parameter of interest. If the variance of an unbiased estimator can achieve the CRLB, it is then said to be an efficient estimator\cite{book:Kay97,arvind_ho}. 

\subsection{CRLB Derivation Using Gaussian PMF Approximation}

Next, we present the expression of CRLB by considering the PMF presented in \eqref{eq:PMF_normal}.
\begin{thm}
Let a UAV is flying over a cellular network with GBS density $\lambda_{\text{GBS}}$ at a fixed height over a linear trajectory and make $H$ handovers within a time period $T$. If the PMF of the HOC can be expressed as $f_H^{(n)}(h;v)$ as in \eqref{eq:PMF_normal}, then the CRLB of the estimated speed $\hat{v}$ is given by:  
\begin{equation}
    \emph{var}(\hat{v}) \geq \frac{(v\sigma)^2}{(c_1\mu)^2+0.5(c_2\sigma)^2},
    \label{eq:crlb_thm}
\end{equation}
where $c_1$ and $c_2$ are the parameters associated with mean and variance, respectively as defined in \eqref{eq:PMF_param_var_normal}.
\end{thm}
\begin{proof}
By definition of CRLB, we know that 
\begin{equation}
     \text{var}(\hat{v}) \geq \frac{1}{I(v)},
\label{eq:crlb_fisher}
\end{equation}
where $I(v)$ is the \textit{Fisher Information} and can be expressed as:
\begin{equation}
I(v)=\mathop{\mathbb{E}}\Bigg [\bigg(  \frac{\partial \log f_H^{(n)}(h;v)}{\partial v}   \bigg )^2 \Bigg].
\end{equation}
Here, $\mathop{\mathbb{E}}[\cdot]$ is the expectation operator with respect to $H$. Consider the PMF approximation $f_H^{(n)}(h)$  in \eqref{eq:PMF_normal} which can be represented as a Gaussian distribution, $H \sim \mathcal{N(\mu,\sigma^2)}$, where $\mu$ and $\sigma^2$ were given in \eqref{eq:PMF_param_var_normal}, respectively. The Fisher information for the general Gaussian observations is then given by~\cite[Section~3.9]{book:Kay97}:
\begin{equation}
    \begin{split}
        I(v)&= \bigg(\frac{\partial \mu}{\partial v}\bigg)^2\frac{1}{\sigma ^2}+\frac{1}{2(\sigma ^4)}\bigg(\frac{\partial \sigma^2}{\partial v}\bigg)^2\\
        &=\frac{(c_1a_1\lambda_{\text{GBS}}^{b_1}T^{c_1}v^{c_1-1})^2}{\sigma^2}+\frac{(c_2a_2\lambda_{\text{GBS}}^{b_2}T^{c_2}v^{c_2-1})^2}{2\sigma^4}\\
        &\overset{(a)}{=} \frac{1}{v^2}\Bigg[\frac{(c_1\mu)^2}{\sigma^2}+\frac{c_2^2}{2}\Bigg],\\
    \end{split}
    \label{eq:fisher_normal}
\end{equation}
where (a) comes from the expressions of $\mu$ and $\sigma^2$. By placing the expression of $I(v)$ in \eqref{eq:fisher_normal}, we can obtain the CRLB for $\hat{v}$ as into \eqref{eq:crlb_thm}.\end{proof}
\subsection{Minimum Variance Unbiased Estimator Analysis}

So far, we have derived the closed-form CRLB expression by considering Gaussian distribution for approximating the HOC PMF. In this sub-section, we will derive a UAV speed estimator $\hat{v}$ that takes the HOC as input in the closed form. Moreover, we will derive the mean and variance of this estimator $\hat{v}$ and show that this is indeed an MVU estimator. 
We first consider the Rao-Blackwell-Lehmann-Scheffe (RBLS) theorem to find the MVU speed estimator \cite[Section~5.5]{book:Kay97}. According to Neyman-Fisher factorization theorem, if we can factorize the PMF $f_H^{(n)}(h;v)$ as~\cite{arvind_ho}
\begin{equation}
    f_H^{(n)}(h;v)=g(\mathcal{F}(h),v)r(h),
\end{equation}
where $g(\cdot)$ is a function that depends on $h$ only by $\mathcal{F}(h)$, then we can conclude that $\mathcal{F}(h)$ is sufficient statistics of $v$\cite[Section~5.4]{book:Kay97}. We can factor the HOC PMF for our case as: 
\begin{equation}
    f_H^{(n)}(h;v)=\underbrace{ \frac{1}{\sqrt{2\pi\sigma^2}}e^{-\frac{(\mathcal{F}(h)-\mu)^2}{2\sigma^2}}}_{g(\mathcal{F}(h),v)}~\times~\underbrace{1}_{r(h)}.
\end{equation}

Hence, $\mathcal{F}(h)=h$ is sufficient statistics of $v$ and $\mathcal{F}(h)$ can be used to find an MVU estimator such that $\hat{v}=s(\mathcal{F})$ is an unbiased estimator of $v$. By inspecting the relationship between the mean of HOC PMF and $v$ as presented in (10), we propose the following estimator for the UAV speed $v$ as: 
\begin{equation}
    \hat{v}=\sqrt[c_1]{\frac{h}{a_1\times \lambda_{\text{GBS}}^{b_1} \times T^{c_1}}}.
    \label{eq:estimator_v}
\end{equation}
We can calculate the mean of this estimator to check whether this is an unbiased estimator as:
\begin{equation}
    \begin{split}
         E(\hat{v})&= \frac{E(h^{\frac{1}{c_1}})}{K_1}\\
        &=\frac{1}{K_1}\sum_{h=0}^{\infty}h^{\frac{1}{c_1}}\frac{1}{\sqrt{2\pi\sigma^2}}e^{-\frac{(h-\mu)^2}{2\sigma^2}},\\
    \end{split}
    \label{eq:mean_v_hat}
\end{equation}
where $K_1=\sqrt[c_1]{a_1\times \lambda_{\text{GBS}}^{b_1} \times T^{c_1}}$. Due to the complexity of the right hand side of~\eqref{eq:mean_v_hat}, it is not possible to obtain a closed-form solution. However, for $c_1 \neq 1$, $E(\hat{v})$ will not be equal to $v$, and hence, the estimator presented in \eqref{eq:estimator_v} is biased. Next, we calculate the variance of this estimator as follows:
\begin{equation}
    \begin{split}
         \text{var}(\hat{v})&= \frac{\text{var}(h^{\frac{1}{c_1}})}{K_1^2}\\
         &= \frac{1}{K_1^2}\big[E(h^{\frac{2}{c_1}})-E(h^{\frac{1}{c_1}})\big]\\
        &=\frac{1}{\sqrt{2\pi\sigma^2}K_1^2}\Bigg[\sum_{h=0}^{\infty}h^{\frac{2}{c_1}}e^{-\frac{(h-\mu)^2}{2\sigma^2}}-\sum_{h=0}^{\infty}h^{\frac{1}{c_1}}e^{-\frac{(h-\mu)^2}{2\sigma^2}}\Bigg].\\
    \end{split}
    \label{eq:var_v_hat}
\end{equation}

Similar to \eqref{eq:mean_v_hat}, we can not obtain a closed-form expression of $\text{var}(\hat{v})$. Hence, we will obtain the mean and variance of the estimator in \eqref{eq:estimator_v} numerically using the semi-analytic expressions in \eqref{eq:mean_v_hat} and \eqref{eq:var_v_hat}, respectively.

Interestingly, for $c_1 = 1$, the mean of the estimator $\hat{v}$ in \eqref{eq:estimator_v} can be expressed as:
\begin{equation}
   E(\hat{v})=\frac{E(h)}{K_2}=\frac{\mu}{K_2}=v,
\end{equation}
 which is an unbiased estimator of $v$. Here, $K_2={a_1\times \lambda_{\text{GBS}}^{b_1} \times T}$ and $\mu=a_1\times\lambda_{\text{GBS}}^{b_1} \times vT=vK_2$. Since RBLS theorem is used to obtain $\hat{v}$, it turns into an MVU estimator for $c_1 = 1$. Next, we derive the variance of $\hat{v}$ to verify whether this is an efficient estimator as follows: 
 
 \begin{equation}
    \begin{split}
          \text{var}(\hat{v})&= \text{var}(\frac{h}{K_2})=\frac{\sigma^2}{K_2^2}\\
        &\overset{(a)}{=}\bigg(\frac{v\sigma}{\mu}\bigg)^2.\\
    \end{split}
    \label{eq:var_v_hat2}
\end{equation}
Here (a) comes from $K_2=\frac{\mu}{v}$. By comparing \eqref{eq:var_v_hat2} with \eqref{eq:fisher_normal}, we can conclude that the variance of MVU estimator in \eqref{eq:var_v_hat2} is greater than the CRLB, and hence, this $\hat{v}$ for $c_1 = 1$ is not an efficient estimator. However, as the GBS density increases, the variance of the MVU estimator becomes closer to the CRLB, which we will show in Section~\ref{sec:simulation}. 

\begin{cor}
For low values of $t_{\rm TTT}$, the parameter $c_1$ becomes $1$, and the estimator $\hat{v}$ becomes an MVU estimator whose variance increases with UAV speed $v$.


\end{cor}

\section{Mobility State Detection}
\label{sec:MSD analysis}
Here, we perform the statistical analysis of MSD by which a UAV flying will be categorized into one of the three mobility states: low, medium, and high as specified in 3GPP LTE Rel-8 specifications~\cite{arvind_ho,mehta_mobility}.
The speed estimator presented in~\eqref{eq:estimator_v} will be used to estimate the UAV speed based on the HOC for a certain measurement duration $T$. Based on this estimator, the UAV will be categorized into one of the three mobility states, namely, low ($S_{\rm L}$), low ($S_{\rm M}$), and high ($S_{\rm H}$). Similar to~\cite{arvind_ho}, the following conditions will be used to determine the MSD:

\begin{equation}
\label{eq:MSD_condition}
  \mathcal{S}= \left \{
  \begin{aligned}
    & S_{\rm L} &&  \textrm{if}~\hat{v} \leq v_{\rm l}, \\
    & S_{\rm M} &&  \textrm{if}~v_{\rm l} \leq \hat{v} \leq v_{\rm u}, \\
    & S_{\rm H} &&  \textrm{if}~\hat{v} \geq v_{\rm u}, \\
  \end{aligned} \right.
\end{equation} 
where, $\mathcal{S} \in \{S_{\rm L}, S_{\rm M}, S_{\rm H}\}$ is the detected mobility state of the UAV, whereas, $v_{\rm l}$ and $v_{\rm u}$ represent the lower and upper speed thresholds, respectively.

\subsection{Mobility State Probabilities}

We define the mobility state probability as the probability of a UAV being categorized into a particular mobility state based on its speed $v$. Based on this definition, we can define the mobility state probabilities as $P(\mathcal{S}=S_{\rm L};v)$, $ P(\mathcal{S}=S_{\rm M};v)$, and $ P(\mathcal{S}=S_{\rm H};v)$, when the mobility state is detected as $S_{\rm L}$, $S_{\rm M}$, $S_{\rm H}$, respectively, for a UAV speed $v$. 

For a given UAV speed $v$, the UAV will make $H$ handovers within a certain $T$. Since $H$ is a random variable, the estimated speed obtained from~\eqref{eq:estimator_v} will also be a random variable. Hence, missed detections and false alarms can occur during the mobility state calculation. The analytic expressions for the mobility state probabilities can be derived from the PMF of $\hat{v}$ as
\begin{equation}
\begin{aligned}
    f_{\hat{v}}(v)&= P(\hat{v}=v)=P\bigg(\sqrt[c_1]{\frac{H}{a_1\times \lambda_{\text{GBS}}^{b_1} \times T^{c_1}}}=v\bigg)\\
    &=P\bigg(H=a_1\times \lambda_{\rm GBS}^{b_1}\times (vT)^{c_1}\bigg)=f_H(h),
\end{aligned}
\end{equation}
where $h=a_1\times \lambda_{\rm GBS}^{b_1}\times (vT)^{c_1}$. Using the approximation with Gaussian distribution as in~\eqref{eq:PMF_normal}, the PMF of $\hat{v}$ can be expressed as:
\begin{equation}
\begin{aligned}
     f_{\hat{v}}(v)&=f_H(h)~\approx f_H^{(n)}(h)\\
     &=\frac{1}{\sqrt{2\pi\sigma^2}}e^{-\frac{(h-\mu)^2}{2\sigma^2}}, \text{for}~h \in \{0,1,2,...\}.
\end{aligned}
\end{equation}
 
Based on the above analysis, we can express the three mobility state probabilities as
\begin{equation}
\begin{aligned}
 P(\mathcal{S}=S_{\rm L};v)&= P(\hat{v} \leq v_{\rm l})~\approx \sum_{h=0}^{h_{\rm l}} f_H^{(n)}(h),\\
 P(\mathcal{S}=S_{\rm M};v)&= P(v_{\rm l} < \hat{v} \leq v_{\rm u})~\approx \sum_{h=h_{\rm l}+1}^{h_{\rm u}} f_H^{(n)}(h),\\
 P(\mathcal{S}=S_{\rm H};v)&= P(\hat{v} > v_{\rm u})~\approx \sum_{h=h_{\rm u}+1}^{\infty} f_H^{(n)}(h),\\
\end{aligned}
\label{eq:mobility_probabilities}
\end{equation}
where $h_{\rm l}$ and $h_{\rm u}$  are the lower and upper handover count thresholds, respectively, and can be expressed as
\begin{equation}
\begin{aligned}
  h_{\rm l}&=\floor*{a_1\times \lambda_{\rm GBS}^{b_1}\times (v_{\rm l}T)^{c_1}}, \\
  h_{\rm u} &= \floor*{a_1\times \lambda_{\rm GBS}^{b_1}\times (v_{\rm u}T)^{c_1}}.
\end{aligned}
\label{eq:hoc_thresholds}
\end{equation}

For a particular speed thresholds $v_{\rm l}$ and $v_{\rm u}$, the associated HOC thresholds can significantly impact the MSD of a UAV. In~\eqref{eq:hoc_thresholds}, we provide the HOC thresholds for MSD theoretically which are associated with specific speed thresholds~\cite{arvind_ho}. In other related works, however, the HOC thresholds for MSD are calculated only heuristically from the HOC statistics and only for few different ground UE velocities and GBS densities~\cite{3gpp.36.839,mehta_mobility,turkka_moblility}. Moreover, to the best of authors' knowledge, the statistical relationship between the HOC and UAV speed $v$, is not considered before. In this work, we have derived a closed-form expression for the optimum HOC thresholds as a function of GBS density $\lambda_{\rm GBS}$, HOC measurement duration $T$,  and the speed thresholds ($v_{\rm l}$ and $v_{\rm u}$).  

\subsection{Probability of Detection and False Alarm}

The probability of detection represents the accurate detection probability of UAV mobility state which can be expressed as:
\begin{equation}
  P_{\rm D}= \left \{
  \begin{aligned}
    &P(\mathcal{S}=S_{\rm L};v), &&  \textrm{if}~\hat{v} \leq v_{\rm l}, \\
    &P(\mathcal{S}=S_{\rm M};v), &&  \textrm{if}~v_{\rm l} \leq \hat{v} \leq v_{\rm u}, \\
    &P(\mathcal{S}=S_{\rm H};v), &&\textrm{if}~\hat{v} \geq v_{\rm u}.
  \end{aligned} \right.
  \label{eq:Probability_of_detection}
\end{equation} 

On the other hand, the probability of false alarm is the false detection probability of UAV mobility state which can be expressed as $P_{\rm FA}=1-P_{\rm D}$.
\section{Simulation Results}
\begin{table}[t]
\centering
\renewcommand{\arraystretch}{1}
\caption {Simulation parameters.}
\label{Tab:Sim_par}
\scalebox{0.99}
{\begin{tabular}{lc}
\hline
Parameter & Value \\
\hline
$P_{\text{GBS}}$ & $46$ dBm  \\ 
$N_t$ & $8$   \\ 
$\theta$ & $6^\circ$ \\ 
$h_{\text{UAV}}$ & $100$ m\\ 
$h_{\text{GBS}}$ & $30$ m\\ 
${\text{f}_c}$ & 1.5 GHz\\ 
$\lambda$\textsubscript{GBS} & $0.5$, $1$, $1.5$, $2$, $2.5$, and $3$ per $\text{km}^2$\\ 
$v$ & $10 , 30, 60, 90,$ and $120$~km/h\\
$t_{\rm MG}$ & $40$~ms\\
\hline
\end{tabular}}
\end{table}
\label{sec:simulation}
In this section, we will validate the HOC PMF approximation using the Gaussian distribution  $f_H^{(n)}(h)$ by plotting its mean square error (MSE) performance. Then we study the variance of the proposed speed estimator $\hat{v}$ with respect to GBS density $\lambda_{\text{GBS}}$, UAV speed $v$, and HOC measurement time $T$. Simulation parameters are provided in Table~\ref{Tab:Sim_par}. Unless otherwise state, we consider $t_{\rm MG}$~$=40$~ms and $t_{\rm TTT}$~$=0$~ms.
\subsection{PMF Approximation Analysis}
In Section~\ref{sec:stat}, we have introduced the HOC PMF approximation using the Gaussian distribution and the relevant parameters were obtained through curve fitting. Here, we study the accuracy of this approximation by evaluating the MSE between the approximate PMF and the PMF obtained from simulations. The MSE can be expressed as:
\begin{equation}
    \text{MSE}=\frac{1}{L}\sum_{l=1}^L \big[ f_H(l)-f_H^{(n)}(l)\big]^2,
\end{equation}
where $L$ is the number of samples in the PMFs. In Fig.~\ref{fig:MSE_vs_lambda}, we plot the MSE performance with respect to various $v$ and $\lambda_{\text{GBS}}$. Note that here $\lambda_{\text{GBS}}=0.5$ means there is 1 GBS per $2~\text{km}^2$. We can conclude that our approximated PMF matches closely with the actual PMF obtained from the simulations. Other than $v=10$ km/h, the MSE performance tends to decrease slightly with increasing GBS density. This happens due to the slightly erroneous PMF approximation for low UAV speeds. For low $v$, a portion of the left-tail of $f_H^{(n)}$ is not considered due to the non-negative nature of HOC which is shown in~Fig.~\ref{fig:velocity_fitting_distributions}(a). The errors also tend to increase with increasing $v$. This is because the heuristic approximations of $\mu$ and $\sigma^2$ in~\eqref{eq:PMF_param_var_normal} start to deviate from the simulations for higher $v$, which is depicted in Fig.~\ref{fig:mu_lambda_theo_sim}.  


\begin{figure}[t]
\centering
\centering{\includegraphics[width=0.75\linewidth]{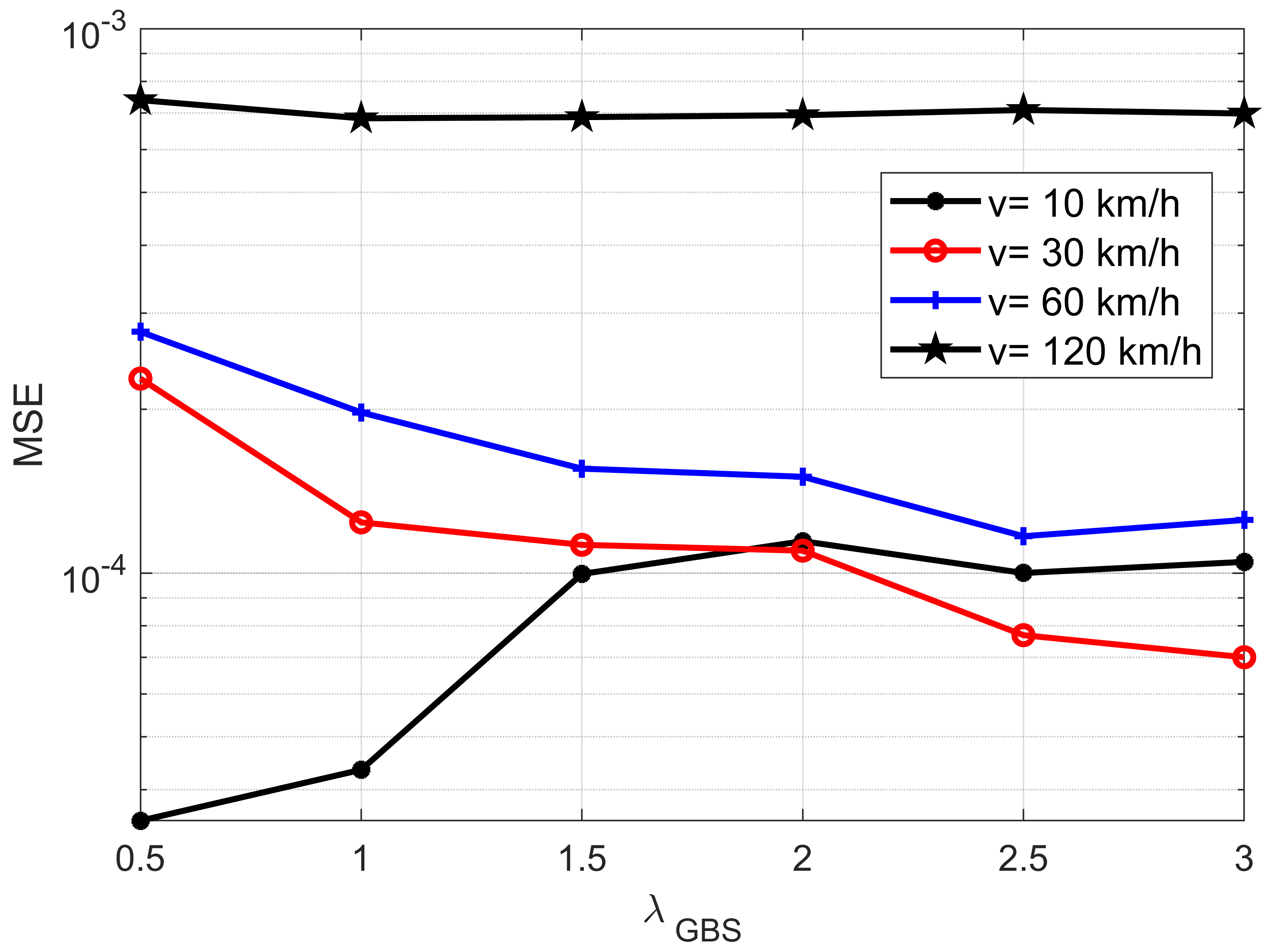}}
    \caption{{MSE versus $\lambda_{\text{GBS}}$ for different $v$ and $T=12$~s with PMF approximation using the Gaussian
distribution.}}
    \label{fig:MSE_vs_lambda}
\end{figure}  

\subsection{CRLB Analysis}
\begin{figure}[t]
\centering
	\subfloat[]{
			\includegraphics[width=.75\linewidth]{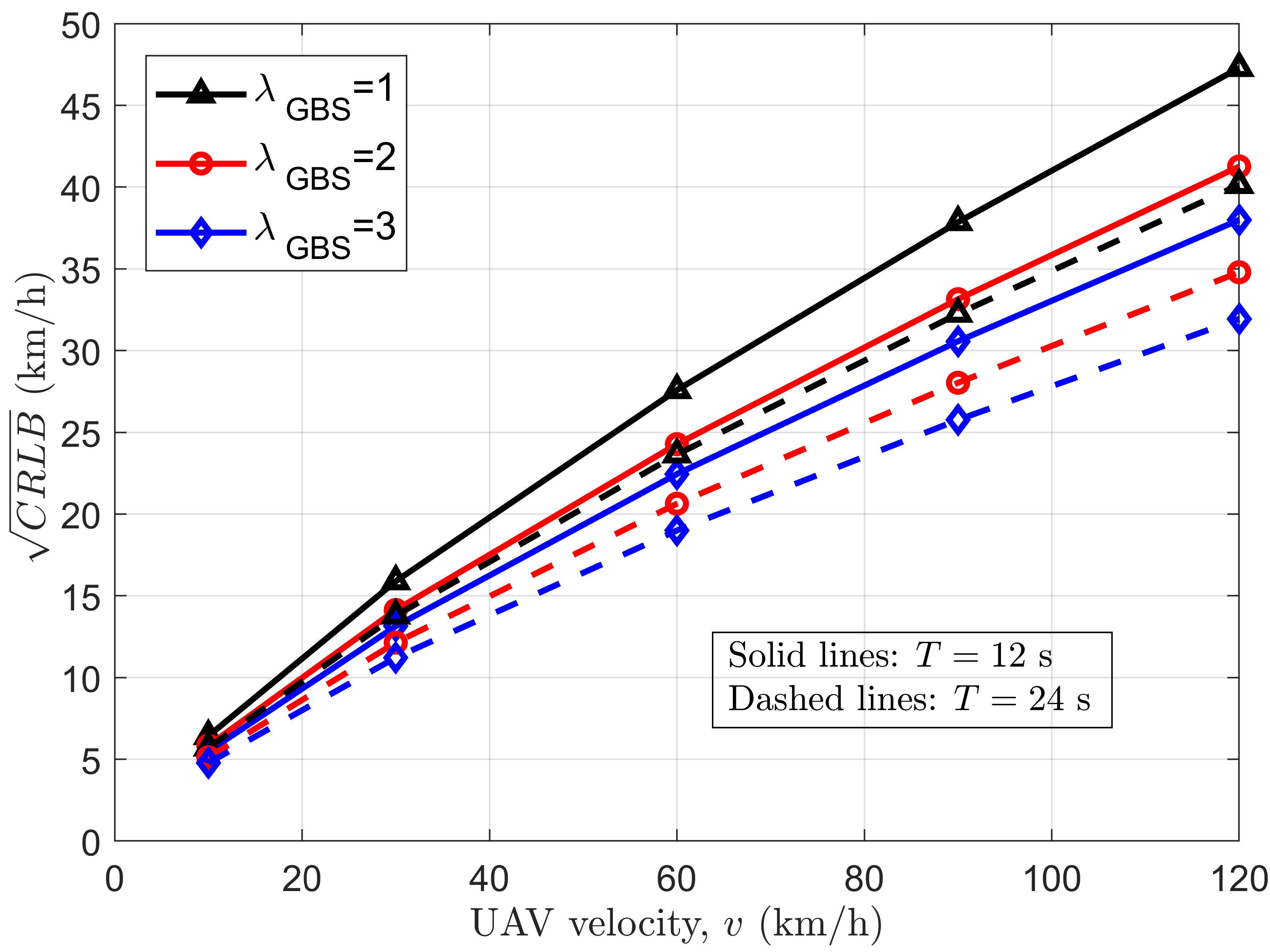}} 
			\hfill
		\subfloat[]{
			\includegraphics[width=.75\linewidth]{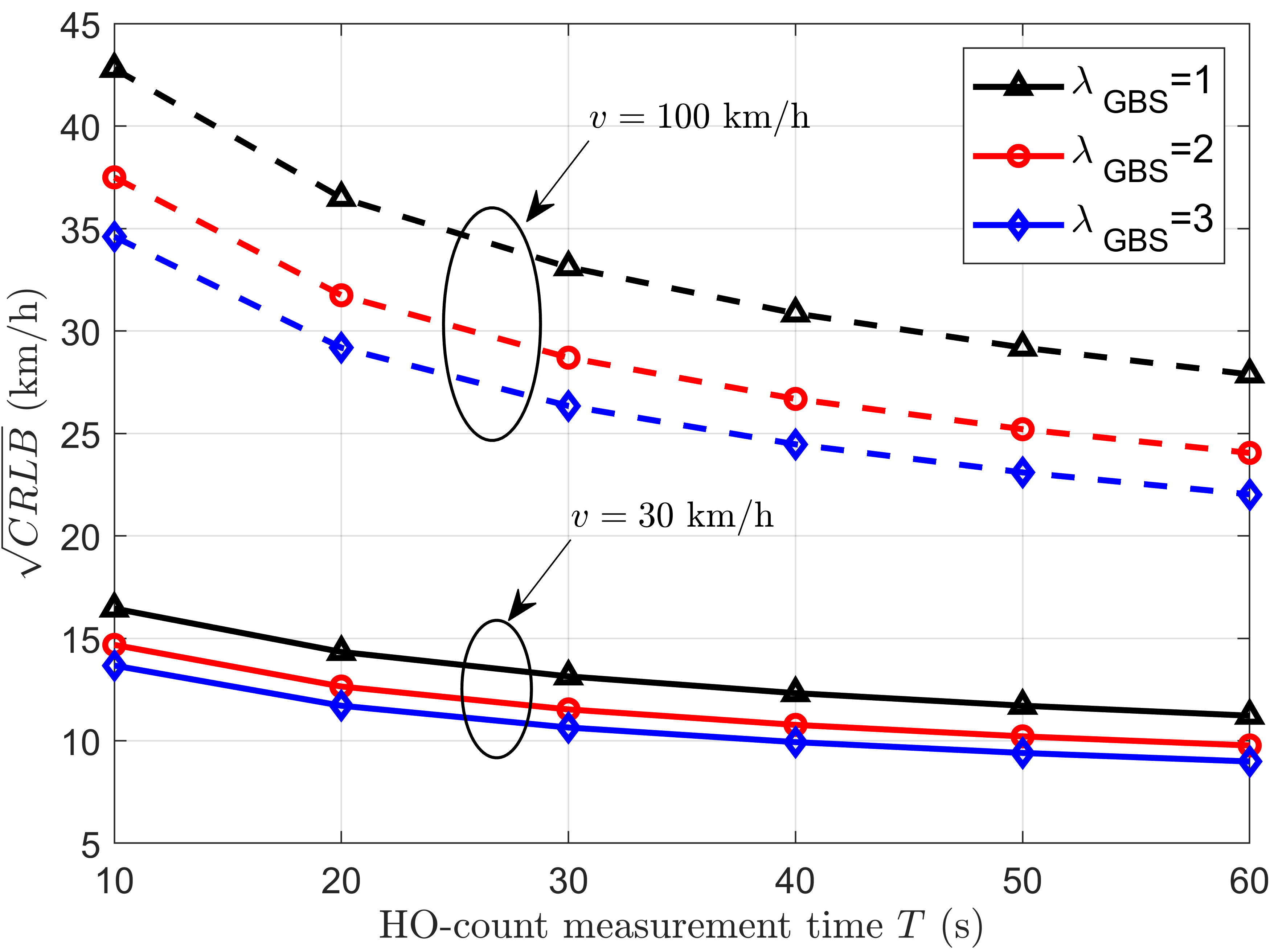}} 
    \caption{{(a) CRLB versus $v$ for different $\lambda_{\text{GBS}}$, and for $T=12$~s and $T=24$~s. (b) CRLB versus $T$ for different $\lambda_{\text{GBS}}$, and for $v=30~\text{km}/\text{h}$ and $v=100~\text{km}/\text{h}$.}}
    \label{fig:CRLB_vs_velocity_time}
\end{figure}
In Fig.~\ref{fig:CRLB_vs_velocity_time}(a), we plot the square root of the CRLB for the UAV speed estimator $\hat{v}$ with respect to GBS density $\lambda_{\text{GBS}}$. The CRLB of the proposed estimator decreases with increasing $\lambda_{\text{GBS}}$ and increases with $v$ as in \eqref{eq:crlb_thm}. This is in line with Fig.~\ref{fig:uav_velocity_PMF_empritical}, where the PMFs are closely spaced with each other for low $\lambda_{\text{GBS}}$, making it difficult to distinguish between different $v$. The square root of the CRLB increases with increasing UAV speed $v$ as shown in~\eqref{eq:crlb_thm}. For available commercial UAVs with maximum speed of $68$ km/h~\cite{drone_speed}, our proposed simple estimator provides RMSEs of $26$ km/h and $21$ km/h for $\lambda_{\text{GBS}}=1$ and $\lambda_{\text{GBS}}=3$, respectively, with $T=24$~s. Fig.~\ref{fig:CRLB_vs_velocity_time}(a) also shows that for $v=120$ km/h, the speed can be estimated with RMSE errors of $47$ km/h and $38$ km/h for $\lambda_{\text{GBS}}=1$ and $\lambda_{\text{GBS}}=3$, respectively. For UAV velocities less than $20$ km/h, the speed can be estimated with RMSE less than $10$ km/h for both of the HOC measurement durations.
\begin{figure}[t]
\centering
	\subfloat[Impact of $h_{\text{UAV}}$]{
			\includegraphics[width=.8\linewidth]{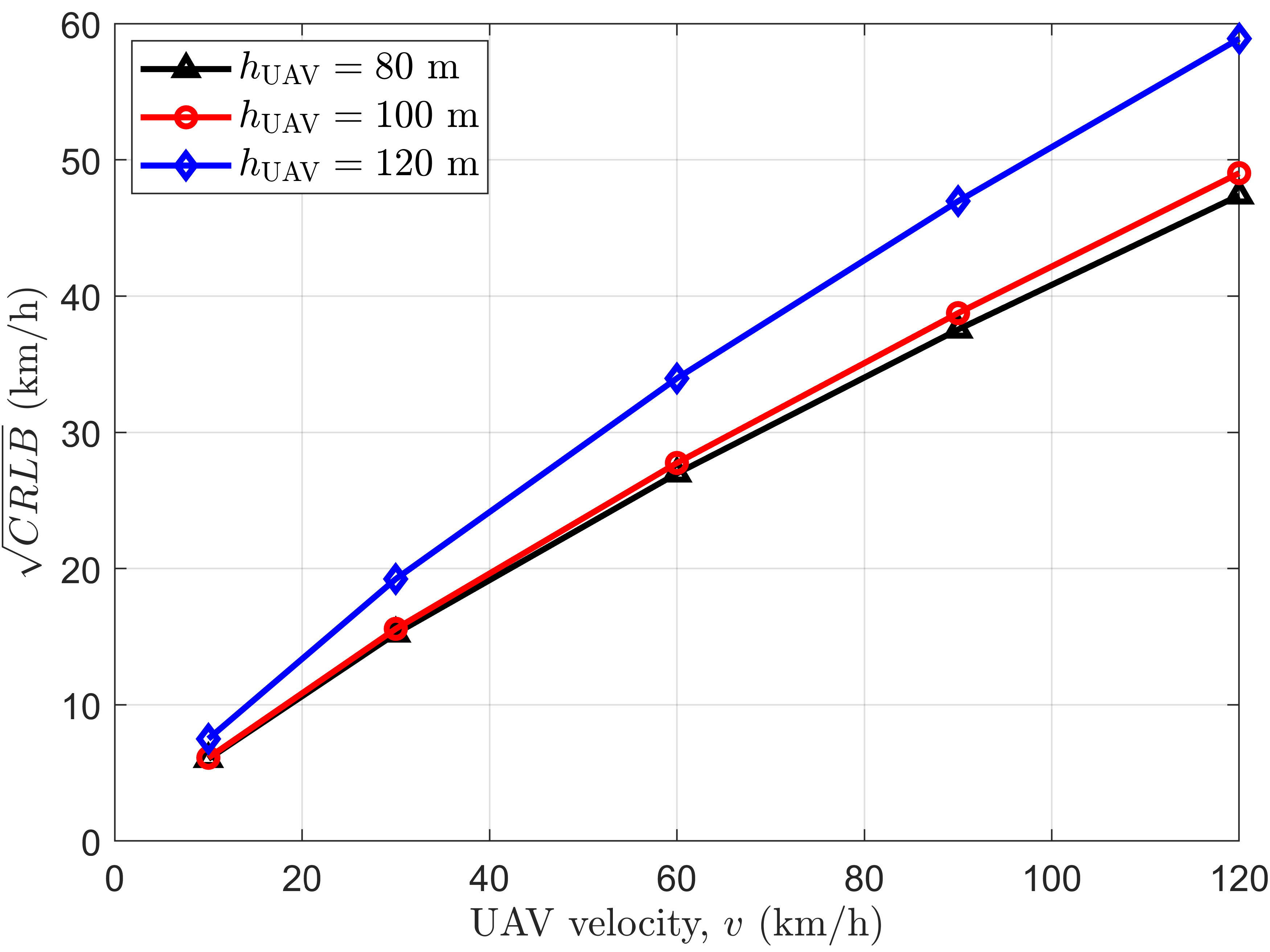}}//
		\subfloat[Impact of TTT]{
			\includegraphics[width=.8\linewidth]{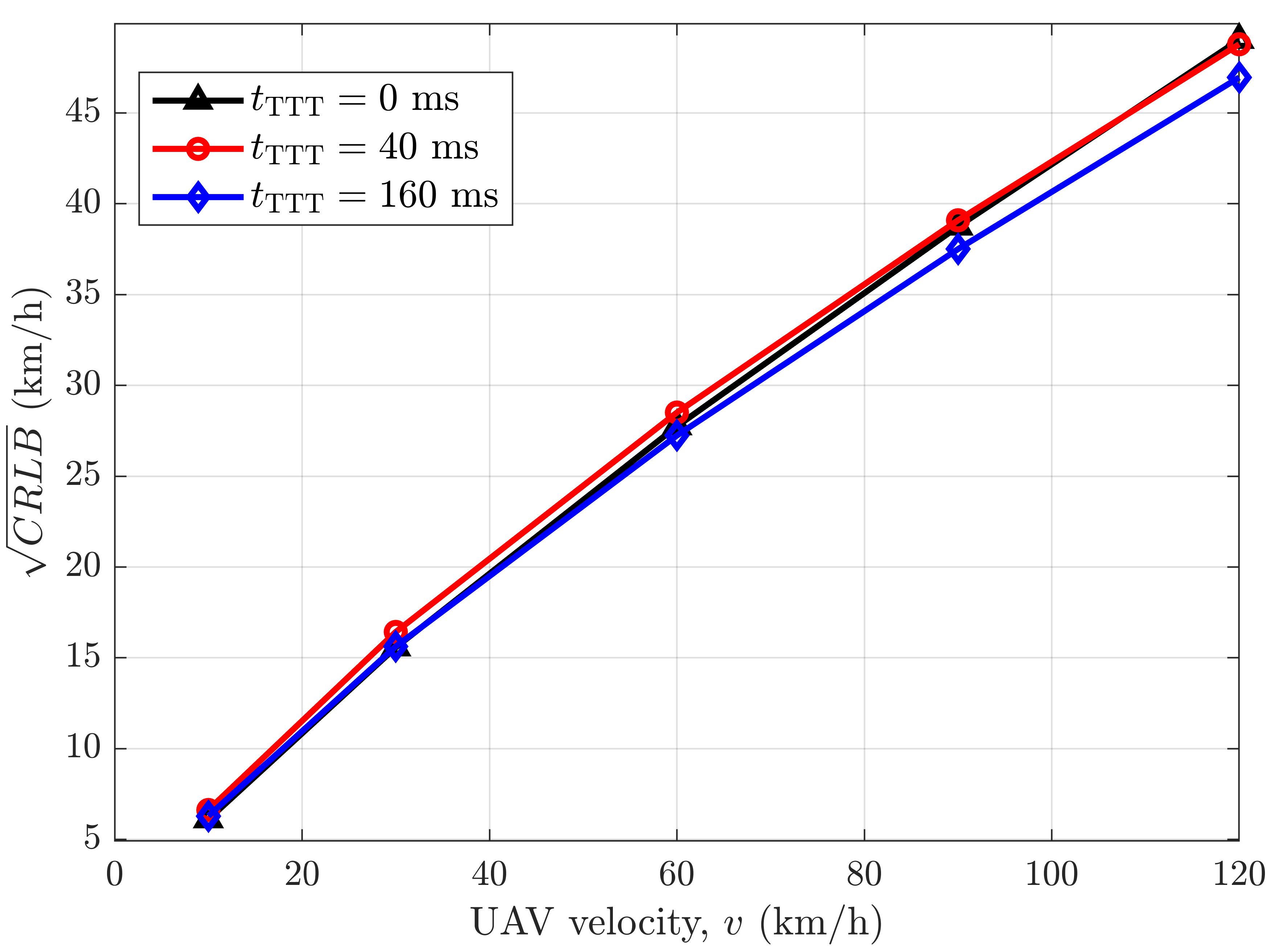}} 
    \caption{{(a) CRLB versus $v$ for different $h_{\text{UAV}}$. (b)  CRLB vs $v$ for different TTT. Here, we consider $\lambda_{\text{GBS}}$=1, $t_{\rm MG}$=$40$~ms, and $T=12$~s.}}
    \label{fig:CRLB_vs_velocity_height_ttt}
\end{figure}

In Fig.~\ref{fig:CRLB_vs_velocity_time}(b), we show the impact of HOC measurement time $T$ on the CRLB of our proposed speed estimator. For a given $\lambda_{\text{GBS}}$ and $v$, CRLB decreases with increasing $T$. As expected, a higher speed provides lower accuracy for the speed estimation. Overall, a longer HOC measurement window will provide better speed estimation since HOCs will be more distinguishable for various $v$ if we allow more time to count handovers made by the UAV traveling on a linear trajectory. Hence, there exists a trade-off between the rapidness and accuracy of the estimated speeds which is also evident in Fig.~\ref{fig:CRLB_vs_velocity_time}(a).

Fig.~\ref{fig:CRLB_vs_velocity_height_ttt} shows the impact of UAV height $h_{\text{UAV}}$ and $t_{\rm TTT}$ on the CRLB performance of the speed estimation with $\lambda_{\text{GBS}}=1$ and $T=12$~s. In Fig.~\ref{fig:CRLB_vs_velocity_height_ttt}(a), we can see that the increasing UAV heights tend to show higher RMSE error for speed estimation. This is because, the UAV makes fewer handovers with increasing $h_{\text{UAV}}$ which makes the HOC based speed estimation erroneous. For UAV speed upto $60$~km/h, the RMSE performances are quite similar for $h_{\text{UAV}}=80$~m and $h_{\text{UAV}}=100$~m. On the other hand, Fig.~\ref{fig:CRLB_vs_velocity_height_ttt}(b) shows the imapct of $t_{\rm TTT}$, from which we can conclude that \emph{the RMSE performances do not vary significantly with respect to $t_{\rm TTT}$ for a UAV speed less than $60$~km/h.}

The variance of the speed estimator given in~\eqref{eq:estimator_v} and its MVU estimator version for low TTT values are shown in Fig.~\ref{fig:CRLB_vs_var(v)}. After observing Fig.~\ref{fig:CRLB_vs_var(v)}(a), we can conclude that the variance of the biased estimator is higher than the CRLB for all of the GBS densities. The gap between the plots tends to decrease with increasing GBS density. As stated earlier, the estimator turns into an MVU estimator for $t_{\rm TTT}$$=0$~ms and $t_{\rm TTT}$$=40$~ms when the parameter $c_1$ equals to $1$. Even though the derived MVU estimator for low TTT values is not an efficient one, the plots of the variance match closely with those of the  CRLB, especially for low UAV velocities.   
\begin{figure}[t]
\centering
	\subfloat[$t_{\rm TTT}$~$=0$ ms]{
			\includegraphics[width=.8\linewidth]{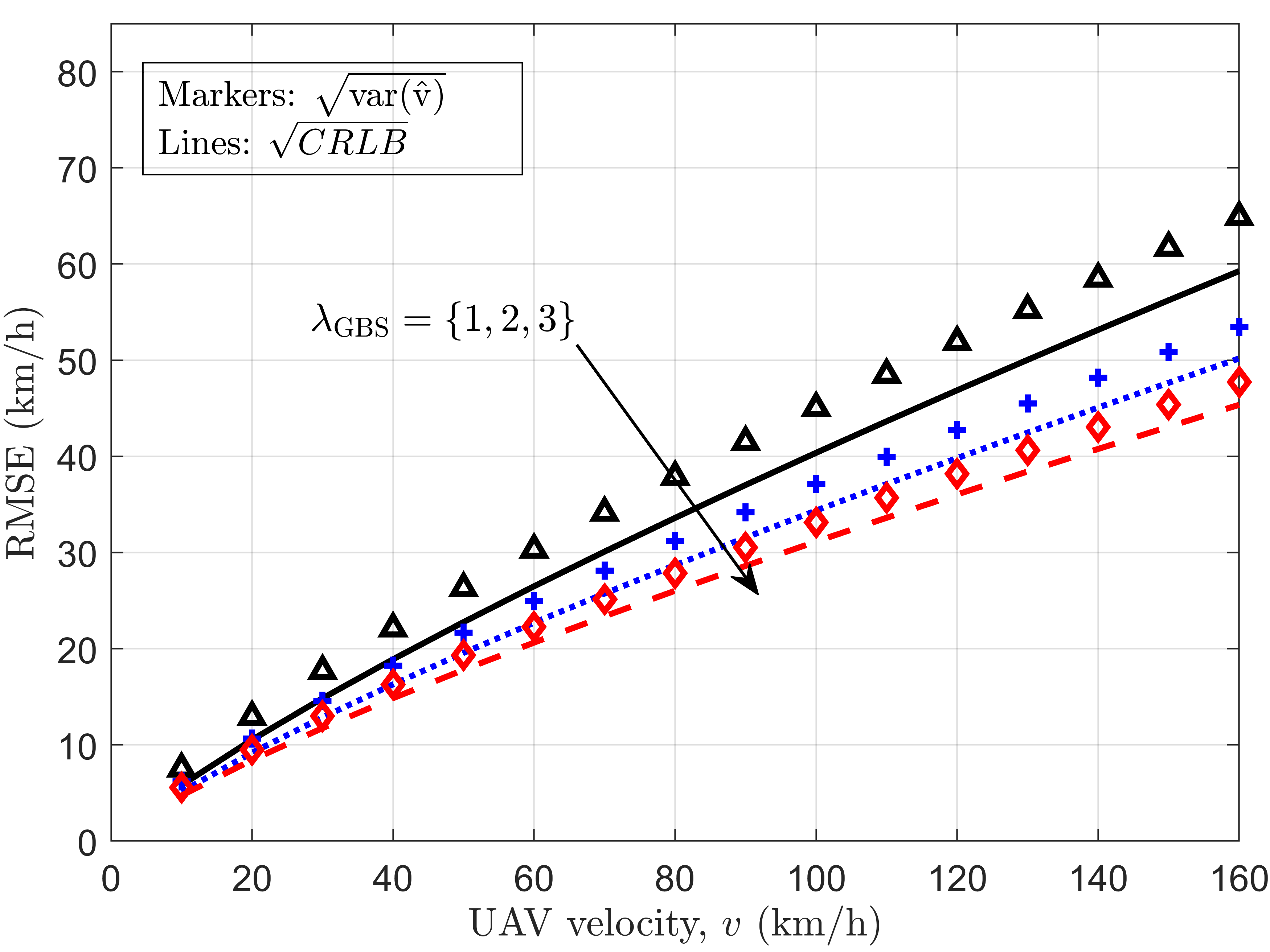}}\\
		\subfloat[$t_{\rm TTT}$~$=160$ ms]{
			\includegraphics[width=.8\linewidth]{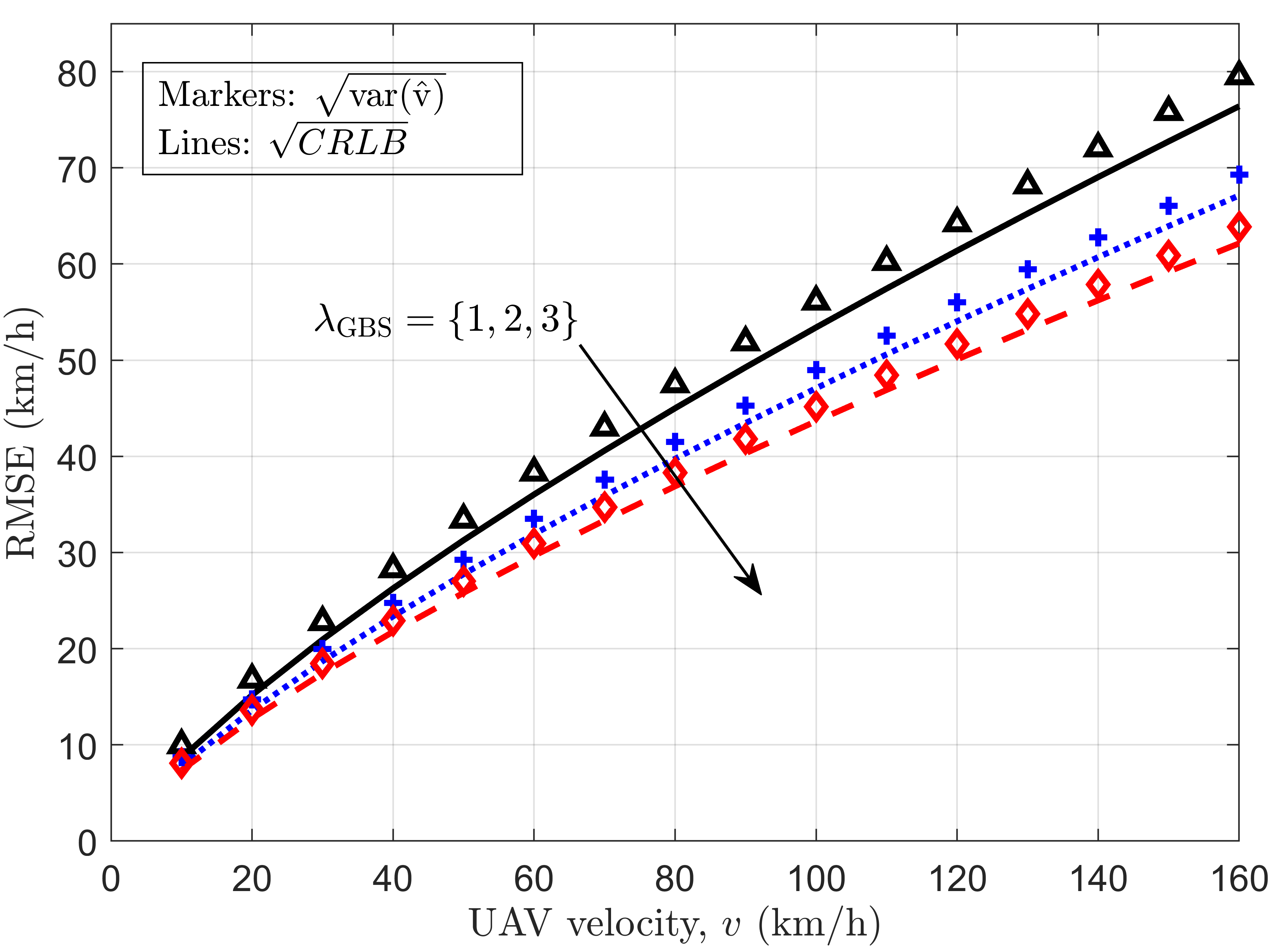}} 
    \caption{{Comparison between the CRLB and the proposed estimator with respect to UAV speed $v$ for different TTT values with $\lambda_{\text{GBS}}=1$.}}
    \label{fig:CRLB_vs_var(v)}
\end{figure}
\subsection{Mobility State Probabilities and Probabilities of Detection}
For obtaining the mobility state of a flying UAV, the service provider can set the speed thresholds $v_{\rm l}$ and $v_{\rm u}$ according to the regulations and requirements. In this work, we set the values of $v_{\rm l}$ and $v_{\rm u}$ to $40$~km/h and $80$~km/h, respectively. We obtain the analytic plots by using \eqref{eq:mobility_probabilities}-\eqref{eq:hoc_thresholds}. In Fig.~\ref{fig:MSD_vs_velocity}, we plot the mobility state probabilities for different $v$ and $\lambda_{\text{GBS}}$. After carefully observing the figures, we can conclude that the slopes of curves slightly increase during their transitions with higher GBS density. In other words, larger GBS densities result in better MSD accuracy. 

\begin{figure}[t]
\centering
	\subfloat[]{
			\includegraphics[width=.75\linewidth]{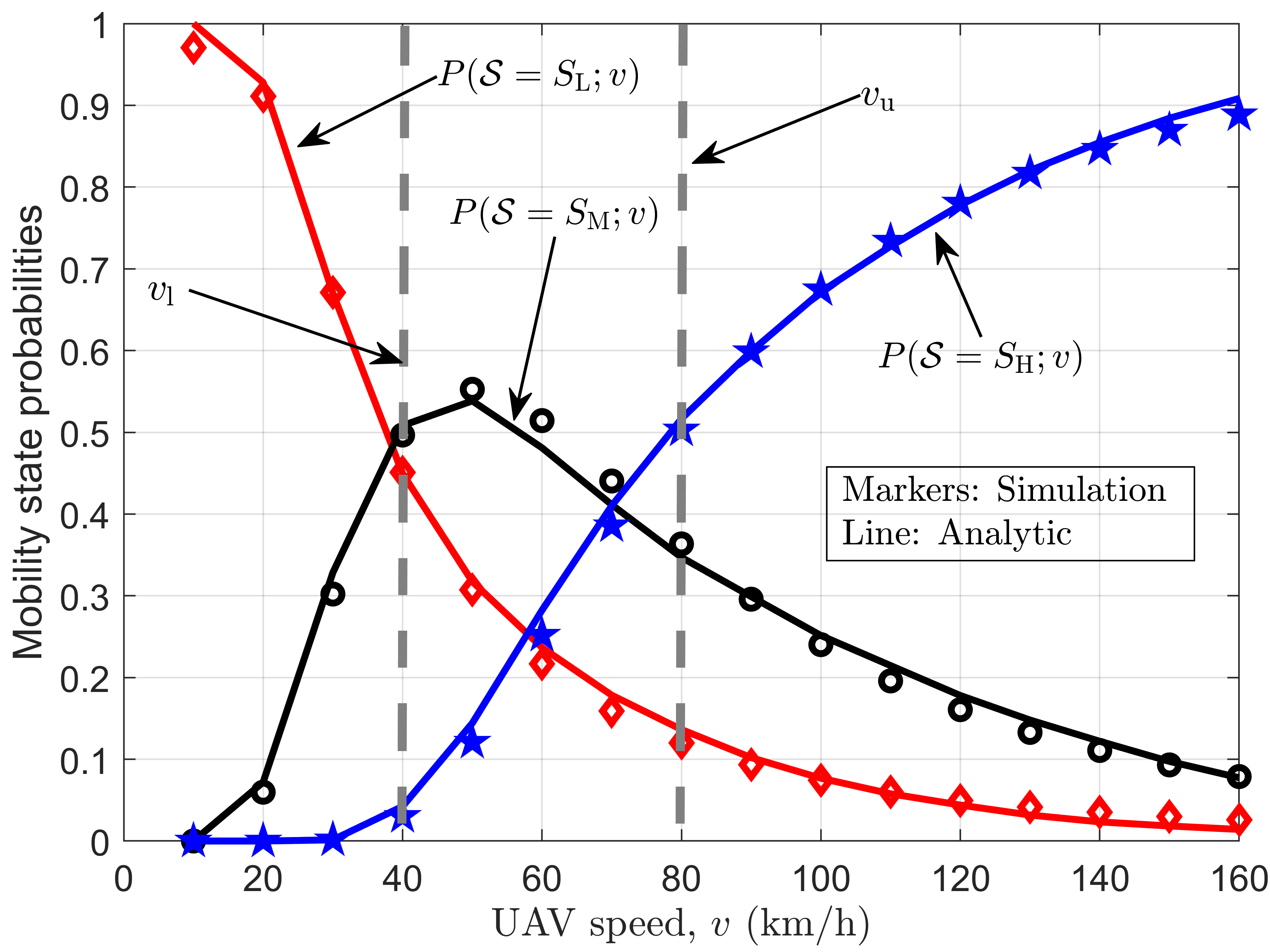}} 
			\hfill
		\subfloat[]{
			\includegraphics[width=.75\linewidth]{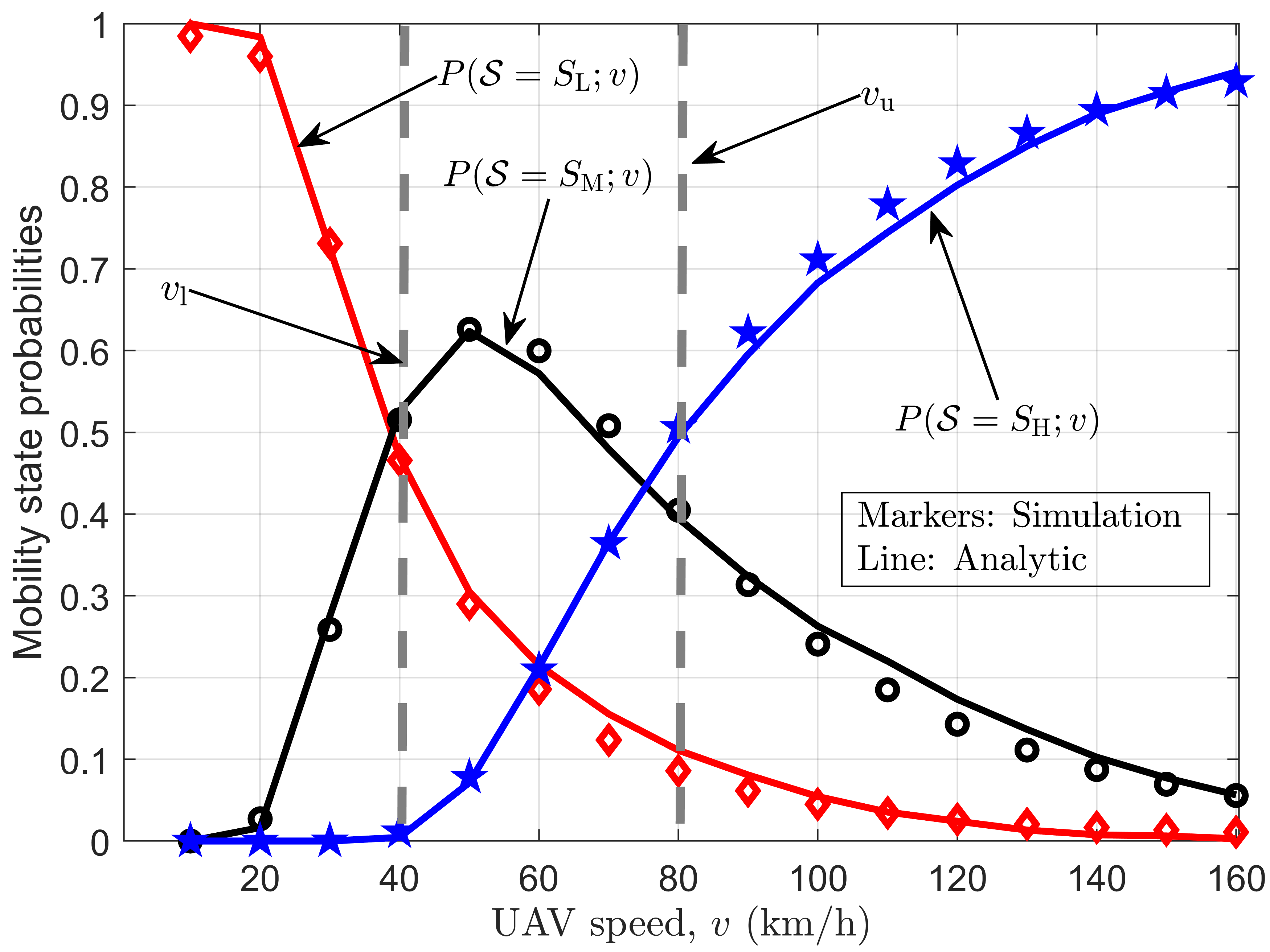}} 
    \caption{{Mobility state probabilities with respect to UAV speed $v$ for $T=12$~s, $v_{\rm l}=40$~km/h, and $v_{\rm u}=80$~km/h; (a) $\lambda_{\text{GBS}}=1$; (b) $\lambda_{\text{GBS}}=3$.}}
    \label{fig:MSD_vs_velocity}
    \end{figure}
    
\begin{figure}[t]
\centering
	\subfloat[]{
			\includegraphics[width=.75\linewidth]{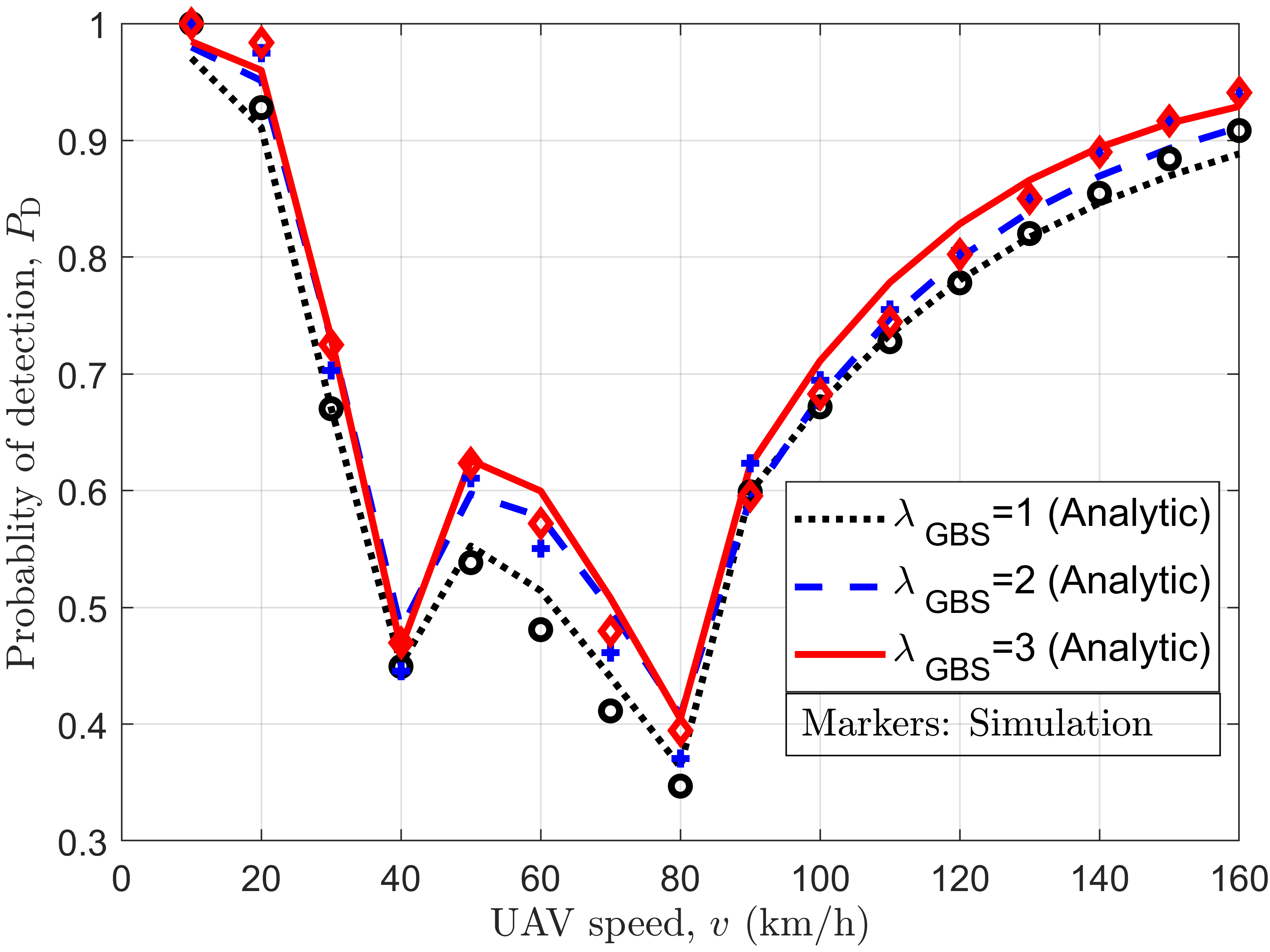}} 
			\hfill
		\subfloat[]{
			\includegraphics[width=.75\linewidth]{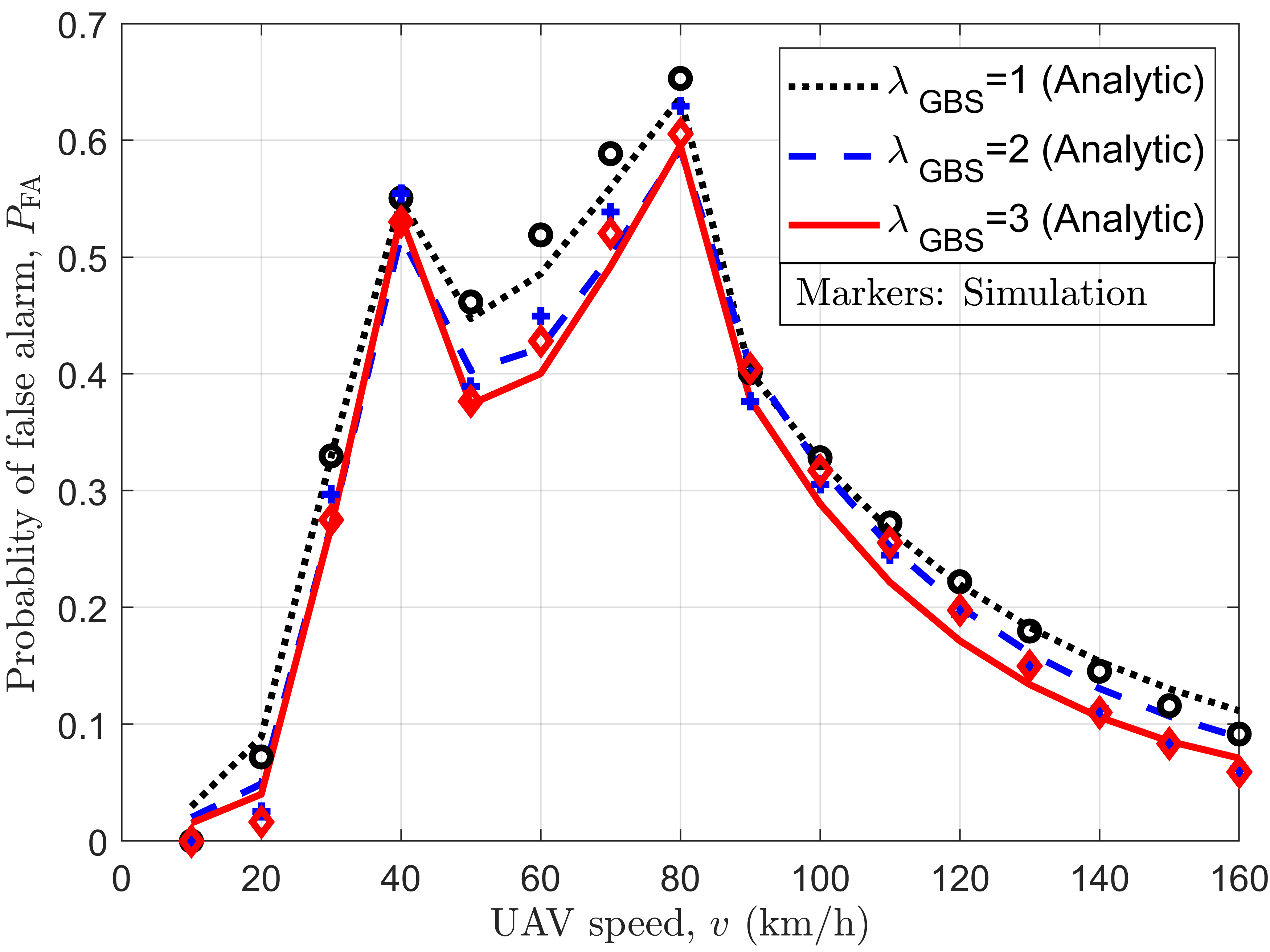}} 
    \caption{{(a) Probability of detection $P_{\rm D}$. (b) Probability of false alarm $P_{\rm FA}$.}}
    \label{fig:POD_PFA_vs_velocity}
    \end{figure}
    
We report the probabilities of detection and false alarm for different $\lambda_{\text{GBS}}$ in Fig.~\ref{fig:POD_PFA_vs_velocity}. The probability of detection decreases and the probability of false alarm increases as the UAV speed approaches one of the speed thresholds. On the other hand, the probability of detection is high for low and high UAV velocities. Note that the analytic probabilities of detection are obtained using~\eqref{eq:Probability_of_detection}.

In Fig.~\ref{fig:POD_vs_threshold}. we plot the average probability of detection $P_{\rm D}$ for different combinations of the HOC thresholds $h_{\rm l}$ and $h_{\rm u}$ such that $h_{\rm u}>h_{\rm l}$. Here, we consider $\lambda_{\text{GBS}}=1$ and $T=12$~s to calculate the average $P_{\rm D}$ for UAV speed in the range of $10$ km/h to $160$~km/h. For a given pair of HOC thresholds and network configuration, the corresponding speed thresholds can be obtained using~\eqref{eq:estimator_v}. We can conclude that the average $P_{\rm D}$ shows high values for low values of $h_{\rm l}$ and $h_{\rm u}$. This is because, for low threshold values, most of the higher UAV velocities will be detected correctly as high mobility states which will result into higher average $P_{\rm D}$. Average $P_{\rm D}$ at first decreases with increasing $h_{\rm u}$ and then increases again. This is due to the fact that for low $h_{\rm l}$ and high $h_{\rm u}$ values, most of the UAV velocities will be detected correctly as medium mobility state category. 
\begin{figure}[t]
\centering
\centering{\includegraphics[width=0.85\linewidth]{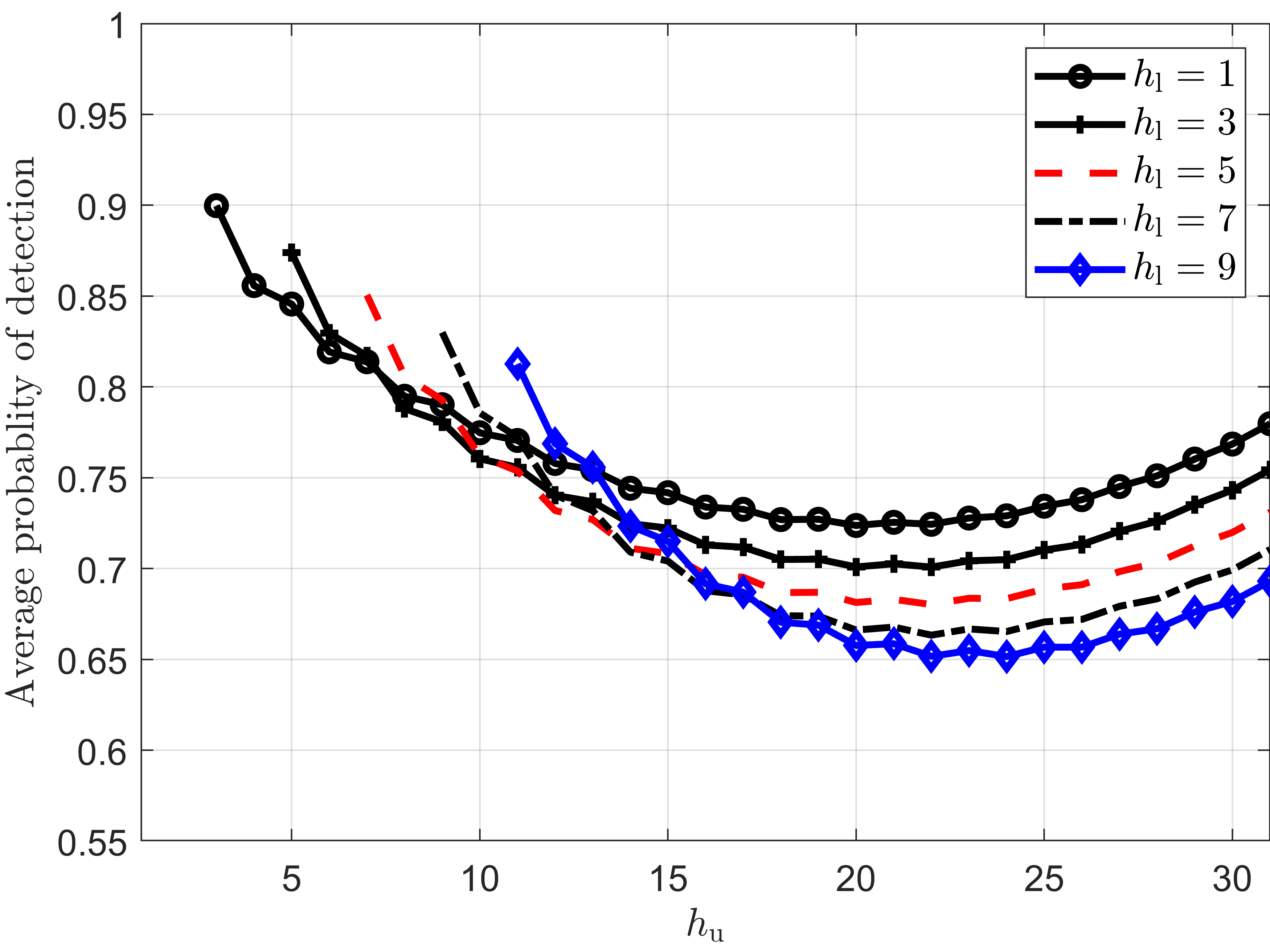}}
    \caption{{Average probability of detection for different combination of $h_{\rm l}$ and $h_{\rm u}$ with  $\lambda_{\text{GBS}}=1$.}}
    \label{fig:POD_vs_threshold}
\end{figure}

\subsection{Mobility State Detection With Variable UAV speed}
We study the functionality of our proposed estimator for variable UAV speed and different HOC measurement duration $T$. We consider a UAV flying in a straight line trajectory at $100$~m altitude while changing its speed with respect to time as shown in Fig.~\ref{fig:variable_velocity_MSD}. The GBS density of the underlying cellular network is considered to be $1$. The values of $v_{\rm l}$ and $v_{\rm u}$ are considered to be $40$ and $80$ km/h, respectively. The speed is estimated with two different schemes, namely, using a discrete window and a sliding window. In the discrete window scheme, the HOC windows of duration $T$ do not overlap with each other throughout the flight time. On the other hand, we consider a HOC window of duration $T$ that slides through the whole flight duration for the sliding window scheme.

The speed estimation performance of both schemes is close to each other where the sliding window technique provides a slightly lower RMSE. However, in the sliding window schemes, the speed is estimated more frequently which leads to higher computational complexity. As expected, the speed estimation process is more accurate for $T=24$~s as shown in Fig.~\ref{fig:variable_velocity_MSD}(b). The corresponding MSD performances for discrete method which are depicted in Fig.~\ref{fig:variable_velocity_MSD} also show higher MSD accuracy for $T=24$. The instances of false alarms tend to decrease with longer $T$. However, a longer HOC interval can also lead to lower estimation accuracy if the UAV changes its speed rapidly during the duration $T$. The sliding window scheme provides the MSD earlier than the discrete method since it estimates the mobility states continuously. 

\begin{figure}[t]
\centering
	\subfloat[]{
			\includegraphics[width=.75\linewidth]{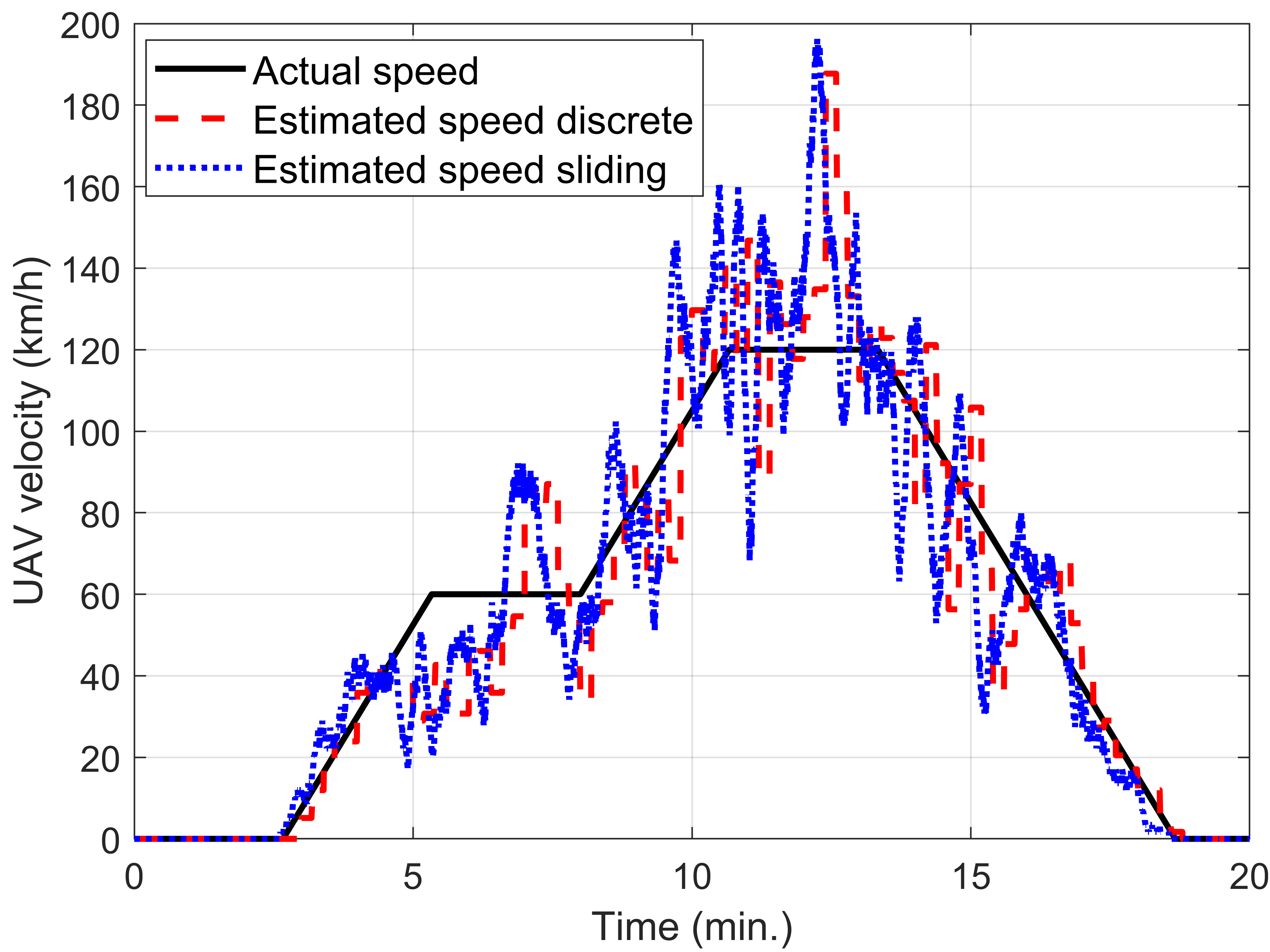}} \\
		\subfloat[]{
			\includegraphics[width=.75\linewidth]{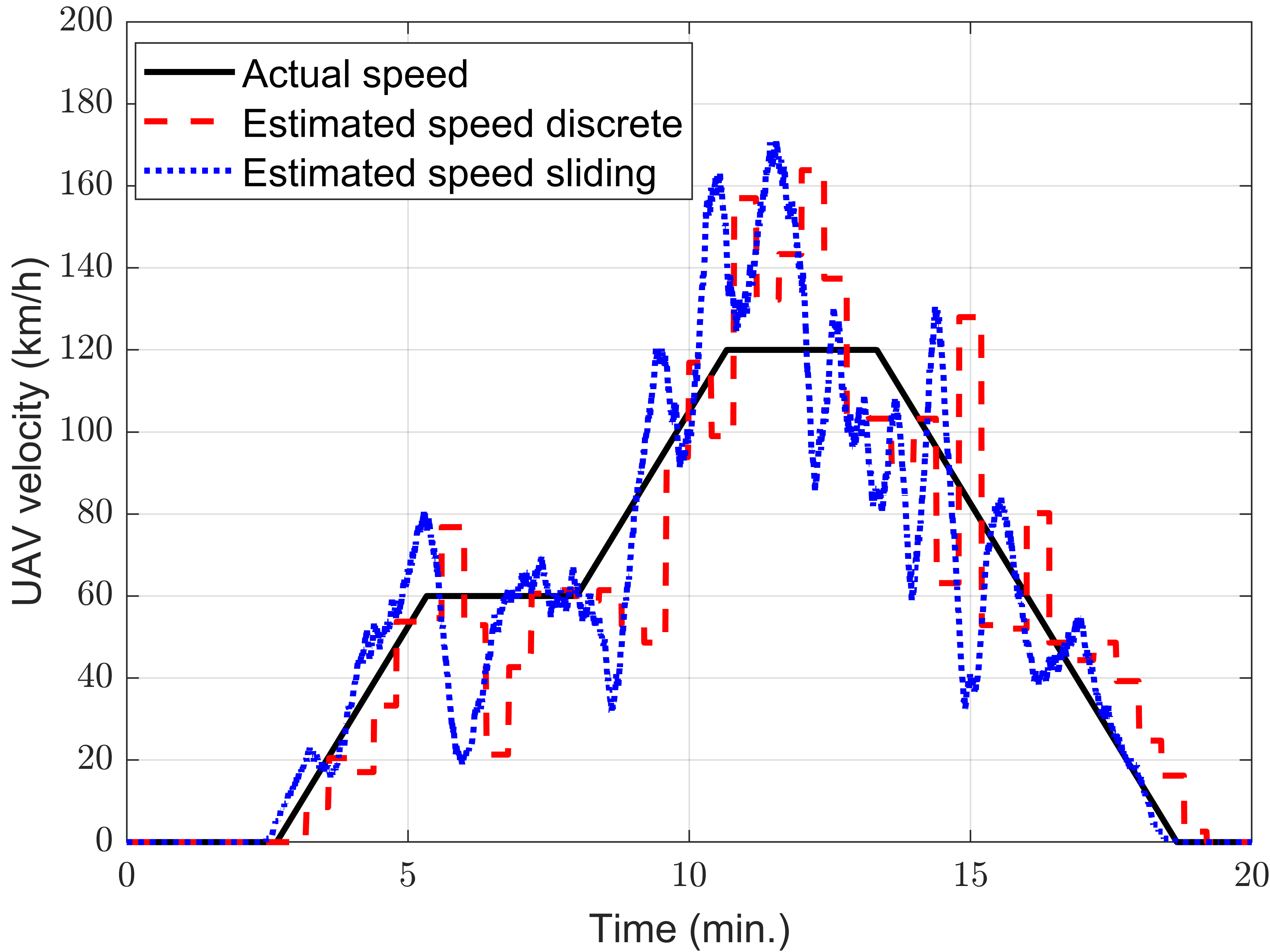}} 
    \caption{{HOC-based speed estimation for variable UAV speed; (a) $T=12$~s; (b) $T=24$~s.}}
    \label{fig:variable_velocity_MSD}
    \end{figure}
    
\begin{figure}[t]
\centering
	\subfloat[]{
			\includegraphics[width=.75\linewidth]{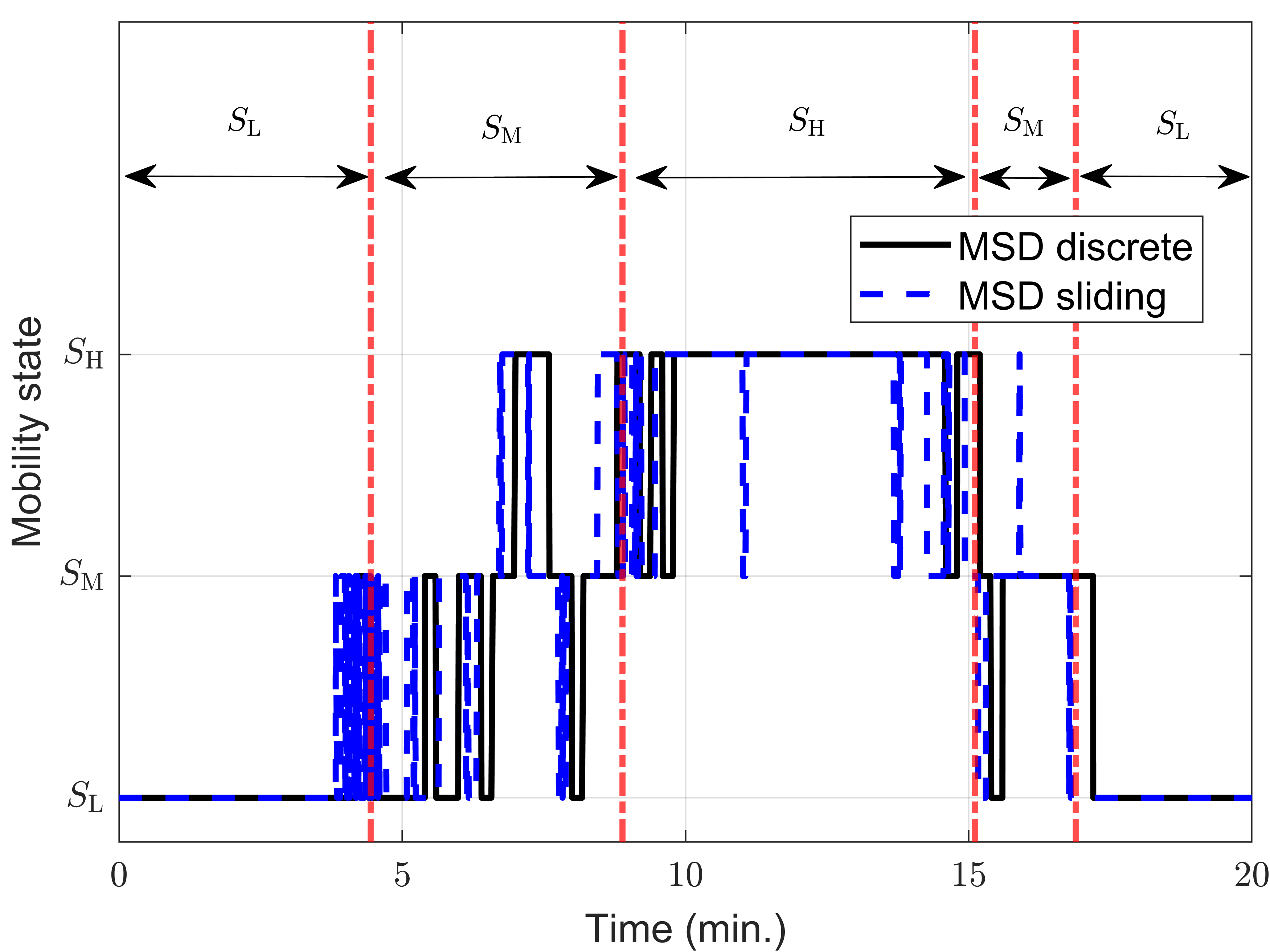}} \\
		\subfloat[]{
			\includegraphics[width=.75\linewidth]{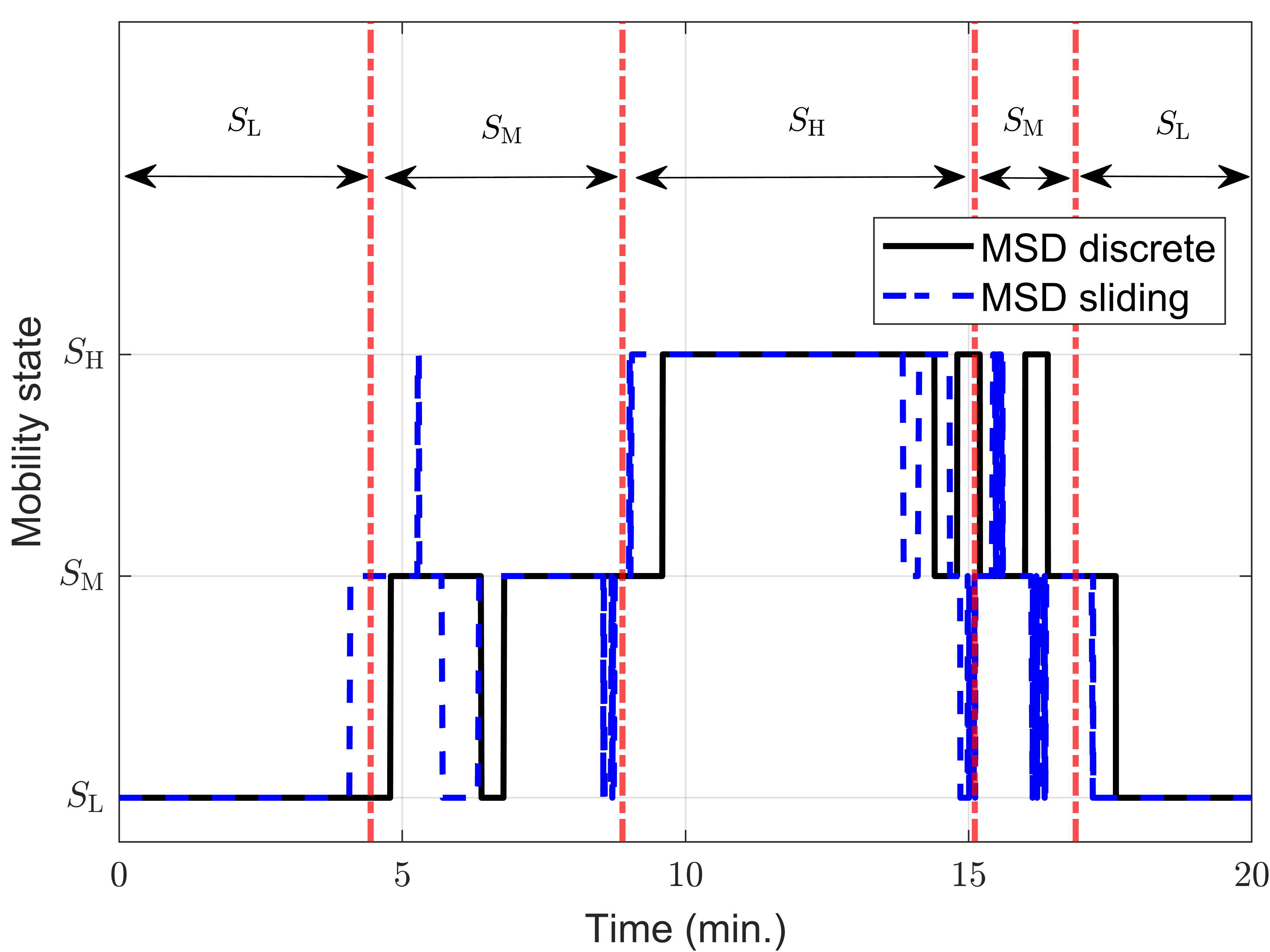}} 
    \caption{{HOC-based mobility state detection for variable UAV speed; (a) $T=12$~s; (b) $T=24$~s.}}
    \label{fig:variable_velocity_MSD2}
    \end{figure}

\section{Concluding Remarks}
\label{sec:Conc}
In this paper, we have approximated the HOC PMFs for aerial UEs in a cellular network using the Gaussian distributions for various heights, handover parameters, and antenna configurations. We have observed that the HOC trends do not change significantly with increasing GBS density, while for high TTT, the UAV tends to make fewer handovers for high UAV speeds. We have also estimated the Gaussian parameters as a function of GBS density, HOC measurement duration, and UAV speed, based on which we have proposed a simple UAV speed estimator. The approximated PMFs have shown high-quality fits with very low MSE. We have also derived the CRLB and provided an MVU estimator analysis from which we have proposed a simple biased estimator. Afterward, this estimator is used to determine the mobility states of the cellular-connected UAV. Our results show that the CRLB of the estimated speed and the mobility state misdetection probability decrease with larger HOC measurement time interval and higher GBS density, which represents the trade-off between the accuracy and the quickness of the speed measurements. We have also shown the effectiveness of our estimator for cases where a UAV flies with different speeds.

Our proposed speed estimation framework can be extended in several ways. For instance, the HOC measurement time window and TTT parameter can be dynamically adjusted based on the past samples of the estimated speeds. Other than considering all the GBSs in the network, the UAV can create a shortlist of only a few GBSs as done in~\cite{galkin2020}. The proposed estimation framework can also be applied in a multi-operator connectivity scheme, where the UAV can associate simultaneously with multiple network providers~\cite{colpaert2021drone}. Though this type of multi-connectivity can provide $99\%$ coverage probability for UAVs, it may increase the network load. In~\cite{drone_simu}, the authors introduced an ns-$3$ based simulator for cellular-connected UAVs with real-world GBS configurations. In~\cite{zeng_CKM}, the authors studied the UAV trajectory optimization problem based on the large-scale knowledge of the propagation environment. Such realistic simulators or radio maps can be used to approximate the realistic HOC PMF or the UAV speed efficiently. In this paper, we did not consider any coverage smoothening technique such as those in~\cite{moin_ICC,chowdhury2021ensuring} to overcome the challenges associated with the scattered GBS association patterns which is another interesting research direction. Our future work will also include estimation of the UAV speed with a three-dimensional flight trajectory, where the UAV can change its altitude in dedicated UAV corridors.




\bibliographystyle{IEEEtran} 
\bibliography{ref}

\end{document}